\documentclass[preprint,12pt]{elsarticle}

\usepackage{amssymb}
\usepackage{amsthm}
\usepackage{amsmath}
\usepackage{algorithm}
\usepackage{algpseudocode}
\usepackage{bm}
\usepackage{subcaption}
\usepackage{graphicx}
\usepackage{multirow}
\newtheorem{theorem}{Theorem}
\theoremstyle{plain}
\newtheorem{definition}{Definition}

\usepackage{url}
\theoremstyle{remark}
\newtheorem*{remark}{Remark}
\newtheorem{strategy}{Strategy}
\usepackage{etoolbox}
\usepackage{hyperref} % put this in the preamble
\makeatletter
\newif\ifblindreview
\blindreviewfalse     % <-- camera-ready / non-blind
\patchcmd{\pprintMaketitle}
  {\footnotesize\itshape\elsaddress\par\vskip36pt}
  {\footnotesize\itshape\elsaddress\par\vskip36pt}{}{}
\begin{document}

\begin{frontmatter}

%% Title, authors and addresses

%% use the tnoteref command within \title for footnotes;
%% use the tnotetext command for theassociated footnote;
%% use the fnref command within \author or \address for footnotes;
%% use the fntext command for theassociated footnote;
%% use the corref command within \author for corresponding author footnotes;
%% use the cortext command for theassociated footnote;
%% use the ead command for the email address,
%% and the form \ead[url] for the home page:
%% \title{Title\tnoteref{label1}}
%% \tnotetext[label1]{}
%% \author{Name\corref{cor1}\fnref{label2}}
%% \ead{email address}
%% \ead[url]{home page}
%% \fntext[label2]{}
%% \cortext[cor1]{}
%% \affiliation{organization={},
%%             addressline={},
%%             city={},
%%             postcode={},
%%             state={},
%%             country={}}
%% \fntext[label3]{}

\title{Model-Agnostic Energy Throughput Control for Range and Lifetime Extension of Electric Vehicles via Cell-Level Inverters}

\ifblindreview
  \author{Anonymous Author(s)}
  \affiliation{organization={Anonymous Institution}}
\else
\author[inst1]{Shida Jiang}
\author[inst1,inst2]{Shengyu Tao\corref{cor1}}
\author[inst3]{Vincent Molina}
\author[inst1]{Junzhe Shi}
\author[inst1]{Scott Moura}

\affiliation[inst1]{organization={Department of Civil and Environmental Engineering},
            addressline={University of California, Berkeley},
            city={Berkeley},
            postcode={94720},
            state={CA},
            country={USA}}
\affiliation[inst2]{organization={Department of Electrical Engineering},
            addressline={Chalmers University of Technology},
            city={Gothenburg},
            postcode={41296},
            country={Sweden}}
\affiliation[inst3]{organization={BMW Group Technology Office USA},
            addressline={2606 Bayshore Pkwy},
            city={Mountain View},
            postcode={94043},
            state={CA},
            country={USA}}

\cortext[cor1]{Corresponding author. Telephone: +460317721689.}
\fi
\begin{abstract}
%% Text of abstract
A conventional electric vehicle (EV) powertrain relies on a centralized high-voltage DC–AC inverter, thereby limiting cell-level control and potentially reducing overall driving range and battery lifetime. This paper studies an H-bridge-based cell-level inverter topology that performs power conversion at the cell level, enabling independent control of individual cells and expanding the design space for battery management. Leveraging these additional degrees of freedom, we propose a model-agnostic energy-throughput control strategy that extends EV range while improving battery-pack lifetime.
Because usable energy (and thus driving range) and lifetime are governed by the cells with the lowest state-of-charge (SOC) and state-of-health (SOH), respectively, the proposed controller preferentially routes energy throughput to healthier cells. Specifically, during charging, it permits cell SOCs to diverge to promote SOH equalization; during discharging, it rebalances SOC to maximize usable capacity under per-cell constraints. The proposed SOC–SOH-aware control strategy is evaluated on two aging models representing lithium manganese oxide and lithium iron phosphate chemistries, using a Tesla Model 3 charge–discharge profile across 14 different parameter settings. Simulations show a 7--38\% improvement in lifetime relative to a conventional SOC-only balancing baseline. More broadly, the results suggest a software-defined pathway to extend EV pack life through routine charging, with minimal reliance on specific degradation models or discharge profiles.

\end{abstract}

%%Graphical abstract
%\begin{graphicalabstract}
%\includegraphics{grabs}
%\end{graphicalabstract}

%%Research highlights
\begin{highlights}
\item Energy Throughput control to prolong battery pack life without sacrificing range
\item No prior knowledge required of detailed degradation models or drive profiles
\item Optimization-based framework for feedback control of multiple cell-level inverters
\item A 7–38\,\% lifetime extension compared with the existing control strategy 
\end{highlights}

\begin{keyword}
%% keywords here, in the form: keyword \sep keyword
Battery management system \sep State-of-charge\sep State-of-health\sep Electric vehicle\sep Cascaded H-bridge\sep Level-shifted pulse-width modulation
%% PACS codes here, in the form: \PACS code \sep code

%% MSC codes here, in the form: \MSC code \sep code
%% or \MSC[2008] code \sep code (2000 is the default)
\end{keyword}
\end{frontmatter}

%% \linenumbers

%% main text
\section{Introduction}
%background
The transportation sector is the largest source of greenhouse gas emissions in the United States, contributing about 27.88\,\% of the total emissions in 2025, according to Climate TRACE \cite{climatetrace}. EVs, which produce zero tailpipe emissions, are therefore receiving increasing attention as a pathway to reduce fossil-fuel use and greenhouse gas emissions.

In an EV powertrain, power electronic converters are required to convert the battery pack's direct current (DC) output into a three-phase alternating current (AC) supply for the traction motor to meet voltage and frequency requirements (and thus speed and torque). Conventionally, as shown in Fig.~\ref{comparison_topology}(a), a high-power DC--AC inverter realizes this conversion \cite{topology2, topology}. However, this traditional topology can face limitations such as motor-current total harmonic distortion (THD) \cite{THD}, efficiency penalties \cite{efficiency}, and a bulky DC-link capacitor \cite{size}. Furthermore, in conventional topologies, the pack is typically used as a whole, and the usable capacity is often constrained by the ``weakest cell'' in the series string \cite{imbalance}. To address these issues, Kandasamy et al.~\cite{topo} proposed a cell-level inverter powertrain. As shown in Fig.~\ref{comparison_topology}(b), the new topology delegates DC--AC conversion to the cell level. These modular units (e.g., H-bridges \cite{topo}, half H-bridges \cite{halfH}, or similar variants) are then connected to the motor. Compared to the conventional topology, the cell-level architecture introduces substantially more degrees of freedom, enabling cells to be charged and discharged at different rates. This distributed power source can reduce motor-current THD and motor losses \cite{THDnew}. The literature also suggests that, when size and weight are translated into cost, the topology can reduce system cost relative to a conventional three-phase, two-level inverter \cite{price}. Additionally, with appropriate SOC balancing, the topology can better utilize residual energy near the end of discharge, thereby extending the EV range.

\begin{figure}[htbp]
\centering
\includegraphics[width=13.5cm]{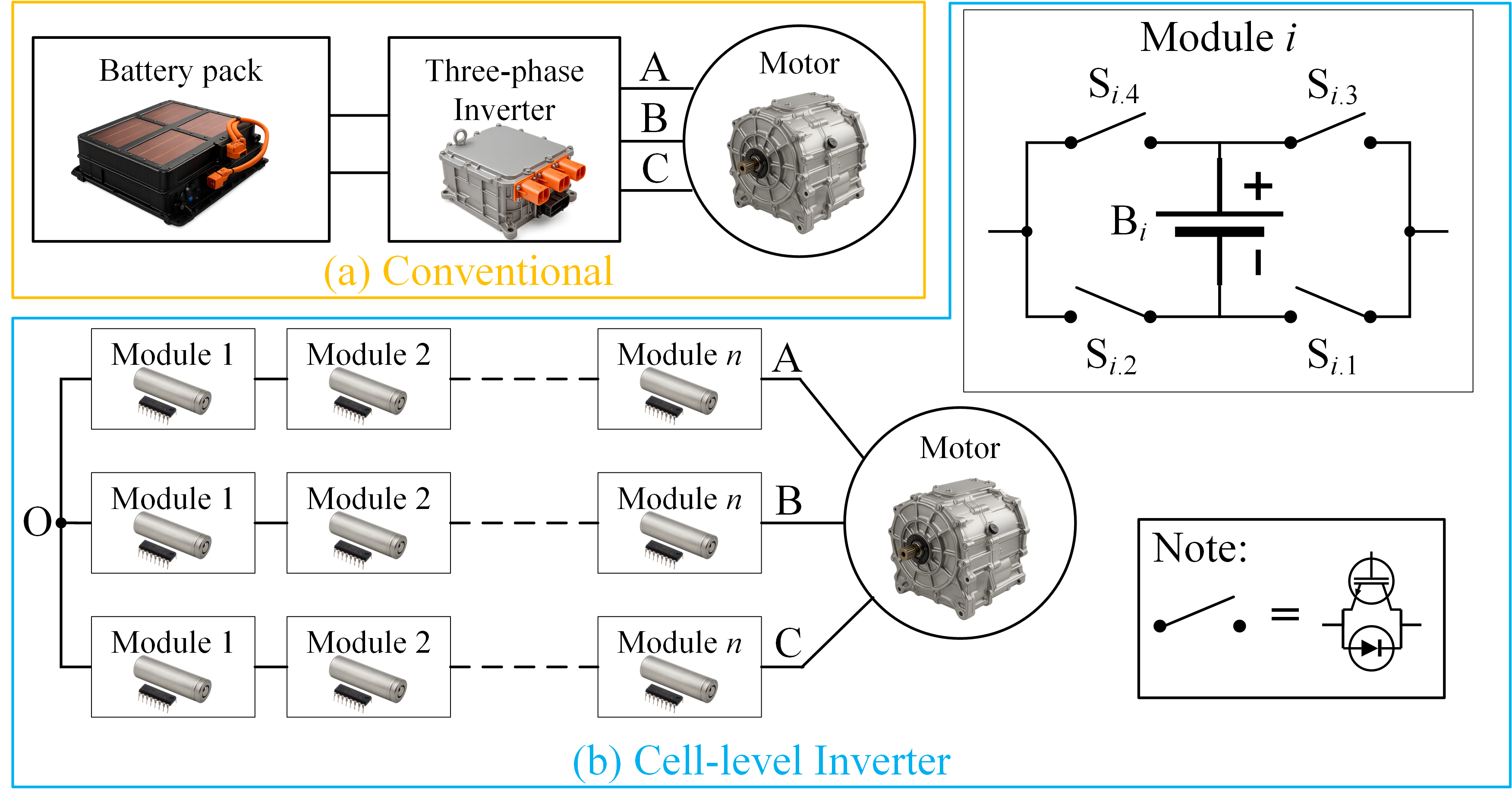}
\caption{Comparison between the (a) conventional inverter powertrain and (b) the H-bridge-based cell-level inverter powertrain. The new topology introduces new degrees of freedom for range and lifetime extension.}
\label{comparison_topology}
\end{figure}

Several studies have investigated SOC balancing strategies for H-bridge-based cell-level inverter topologies. In \cite{topo}, Kandasamy et al. proposed an SOC balancing approach for a three-phase cell-level inverter in EV applications. Their core idea is to prioritize discharging phases and cells with higher SOC while bypassing faulty cells. A key limitation is that the method presumes accurate real-time SOC estimates for every cell, which can impose a substantial computational burden. This requirement is alleviated in \cite{halfH}, where Wu et al. introduced a voltage–SOC balancing strategy that infers SOC differences without performing full SOC estimation for each cell. For charging applications, Gemma et al. \cite{DC_charging1, DC_charging2} developed a constant-current–constant-voltage (CCCV) strategy for a cell-level inverter connected to a DC charger. The strategy performs SOC balancing prior to entering the constant-voltage (CV) stage, after which all cells are charged together during CV. A potential drawback is that balancing is only active during constant-current (CC) charging, so SOC mismatch can re-emerge (or grow) by the end of the CV stage.

Although SOC balancing increases usable energy in EV operation, it is less effective at extending battery-pack lifetime. Experimental evidence shows that lithium-ion cells can age at markedly different rates under identical aging profiles \cite{age1,age2}, particularly in the latter portion of life. In an EV pack, nonuniform thermal conditions further amplify differences in cell-to-cell degradation \cite{heat}. In conventional EV powertrains, pack-level capacity delivery is constrained by the weakest cell \cite{SOHweak}; consequently, the pack reaches end of life (EOL) as soon as any cell reaches EOL. The cell-level inverter architecture mitigates this issue by allowing EOL cells to be bypassed without interrupting normal pack operation. Nevertheless, once failed cells accumulate beyond a threshold (e.g., 10,\% of the total cell count), the pack must still be replaced even if many remaining cells retain substantial health. This leads to under-utilization of available battery life. To the best of our knowledge, however, few studies have developed control strategies for cell-level inverter topologies with the explicit goal of prolonging battery-pack lifetime.

Building on the concept of SOC balancing, battery-pack lifetime can theoretically be extended through SOH balancing, i.e., prioritizing charging and discharging the less-degraded cells. For example, Rehman et al. \cite{SOHbalancing1} proposed assigning different SOC operating limits to cells with different SOH, so that more-degraded cells are cycled less frequently, and overall pack aging is mitigated. However, despite its conceptual appeal, the practical value of this strategy is questionable: it conflicts with SOC balancing and reduces the vehicle’s achievable range. Specifically, SOH balancing tends to produce a highly non-uniform SOC distribution across cells at the end of discharge, leaving usable energy stranded in the pack and thus sacrificing EV range—an outcome that is generally unacceptable in practice.

In this paper, we propose a strategy that prolongs battery life without sacrificing usable capacity (and thus EV range). This objective is particularly challenging because the current must be distributed among cells with fine time resolution, yielding an enormous number of decision variables without further simplification. Moreover, neither the subsequent discharge trajectory nor a reliable battery aging model is assumed to be known. In general, the main contributions of this paper are:
\begin{itemize}
    \item We propose a new SOC–SOH-aware control strategy that extends EV range and battery life by dynamically controlling each cell's Ah throughput. The strategy can extend battery lifetime by about 7--38\,\% across 14 different use cases and two battery chemistries. 
    \item We propose an optimization algorithm and a closed-loop charging and discharging strategy that resolves the contradiction between longer range and longer battery life. 
    \item The proposed strategy does not require prior knowledge about the battery aging model or drive profiles. The control strategy can, therefore, be utilized in a wide range of applications. 
\end{itemize}

The remainder of the paper is organized as follows. Section 2 introduces some key definitions related to the health and remaining capacity of the battery pack. Section 3 introduces the operation principle of the cell-level inverter topology and the feasible region of its control strategy. Section 4 presents our control strategy, which can effectively extend the lifetime of the battery pack without sacrificing the range of the EV. Section 5 presents a comprehensive simulation study to validate the effectiveness of the control strategy. Finally, the conclusions are drawn in Section 6.

\section{Key Definitions}
\subsection{Cell State-of-charge and State-of-health}
In this paper, the remaining capacity and health of cells are defined by two states: SOC and SOH. As described in the introduction, SOC describes the normalized remaining capacity of a cell. Namely,
\begin{equation}\label{SOC}
    SOC=\frac{Q_r}{Q_{max}}\cdot 100\,\%
\end{equation}
where $Q_r$ is the remaining capacity of the cell, and $Q_{max}$ is the present maximum capacity of the cell. SOH describes the normalized present maximum capacity of a cell. According to \cite{SOH},
\begin{equation}\label{cellSOH}
    SOH=\frac{Q_{max}}{Q_{no}}\cdot 100\,\%
\end{equation}
where $Q_{no}$ is the nominal capacity of the cell. 

%%%%
\subsection{End-of-life, Pack SOC, and Pack SOH}\label{end_pack}

Since the main objective of this paper is to extend battery-pack lifetime, we must define battery EOL. In most studies, EOL is defined based on the SOH: when a cell's SOH drops to approximately 70--80\,\%, it is considered to have reached EOL for EV usage \cite{EOL_cell}. However, as noted by Martinez-Laserna et al.~\cite{EOL_pack}, many cells from retired EV packs do not meet this criterion. Indeed, Saxena et al.~showed that most daily driving needs of U.S. drivers can still be satisfied even when SOH is well below 70\,\% \cite{doubt1}. Similar observations were reported in \cite{doubt2,doubt3}, suggesting that a universal 70--80\,\% SOH threshold may be overly conservative for EV applications.

A similar debate exists outside academia. According to the EV battery test procedures manual published by the United States Advanced Battery Consortium (USABC) in 1996 \cite{USABC1996}, a battery or module reaches EOL when either (i) its net delivered capacity is less than 80\,\% of its rated capacity, or (ii) its peak power capability is less than 80\,\% of the rated power at 80\,\% depth of discharge (DOD). However, this percentage-based criterion was removed in the 2015 revision \cite{USABC2015}. In the latest manual, EOL is primarily tied to available capacity and energy density, implying that packs with higher initial energy density may retire at lower SOH. In addition, California Air Resources Board (CARB) regulations require EV manufacturers to prevent battery capacity from deteriorating below 70\,\% during a warranty period of eight years or 100{,}000 miles, whichever occurs first \cite{CARB}. CARB also includes a durability requirement that at least 70\,\% of vehicles in a test group retain at least 70\,\% of the certified range for a useful life of 10 years or 150{,}000 miles, whichever occurs first \cite{CARB3}. While these requirements do not explicitly define pack EOL, they imply that EV battery packs may remain in service until SOH approaches 70\,\%.

Given the considerations above, this paper evaluates multiple thresholds for battery-pack EOL (denoted \(SOH_{\text{EOL}}\)) and discusses their effects on the proposed strategy in Section~\ref{sec_sensitive}. By default, we assume \(SOH_{\text{EOL}}=70\,\%\), motivated by the conventional EOL definition, CARB-related provisions \cite{CARB,CARB3}, and evidence that both capacity fade and resistance growth can accelerate below \(\sim 70\,\%\) SOH \cite{death1,death2}. Similarly, we assume that an individual cell reaches EOL when its SOH, defined in \eqref{cellSOH}, reaches \(SOH_{\text{EOL}}\). Meanwhile, following the definition adopted by CARB \cite{CARB2}, the pack SOH is defined as
\begin{equation}\label{SOHpack}
    SOH_{\text{pack}}=\frac{Q_{\max,\text{pack}}}{Q_{\text{no},\text{pack}}}\cdot 100\,\%,
\end{equation}
where \(Q_{\max,\text{pack}}\) is the maximum usable discharge capacity and \(Q_{\text{no},\text{pack}}\) is the nominal pack capacity (equal to the sum of nominal cell capacities).
For the conventional topology in Fig.~\ref{comparison_topology}(a), \(SOH_{\text{pack}}\) equals the minimum SOH among all cells, since usable pack capacity is limited by the lowest-capacity cell \cite{imbalance}. For the cell-level inverter topology in Fig.~\ref{comparison_topology}(b), when no cell has reached EOL and nominal capacities are identical, \(SOH_{\text{pack}}\) equals the average cell SOH, because the topology can utilize each cell's remaining capacity. Once any cell reaches EOL and is bypassed during discharge, \(Q_{\max,\text{pack}}\) excludes that cell while \(Q_{\text{no},\text{pack}}\) remains unchanged, so \(SOH_{\text{pack}}\) becomes lower than the average SOH. For example, let \(SOH_{\text{EOL}}=70\,\%\) and consider three cells with SOH values 90\,\%, 80\,\%, and 70\,\%. The third cell is at EOL and is bypassed, so \(SOH_{\text{pack}}=(90\,\%+80\,\%)/3=56.7\,\%\), whereas the average SOH across all three cells would be 80\,\%.

Finally, the battery pack SOC is defined as:
\begin{equation}\label{SOCpack}
    SOC_{pack}= \frac{Q_{r, pack}}{Q_{max, pack}}\cdot 100\,\%
\end{equation}
where $Q_{r, pack}$ is the remaining usable capacity of the battery pack and is equal to the sum of the remaining capacity of the cells.

%%%%
\section{Operating Principle and Control of the Cell-Level Inverter Topology}
\subsection{Level-shifted Pulse Width Modulation}
Existing literature has proposed many different types of control methods for the cascaded H-bridge topology shown in Figure \ref{comparison_topology}(b). The most popular ones include phase-shifted pulse width modulation (PSPWM), LSPWM, selective harmonics elimination (SHE), and space-vector pulse width modulation (SVPWM) \cite{PWM1, PWM2}. However, most of these control techniques are designed for balanced DC sources, meaning that each H-bridge is used equally in the long run. Here, however, we want to charge and discharge different cells at different rates so that their SOH can be more balanced. Therefore, the control technique we are especially interested in is LSPWM, and its key idea is illustrated in Figure \ref{LSPWM_fig}.

\begin{figure}[htbp]
\centering
\begin{subfigure}{0.48\textwidth}
\includegraphics[width=\textwidth]{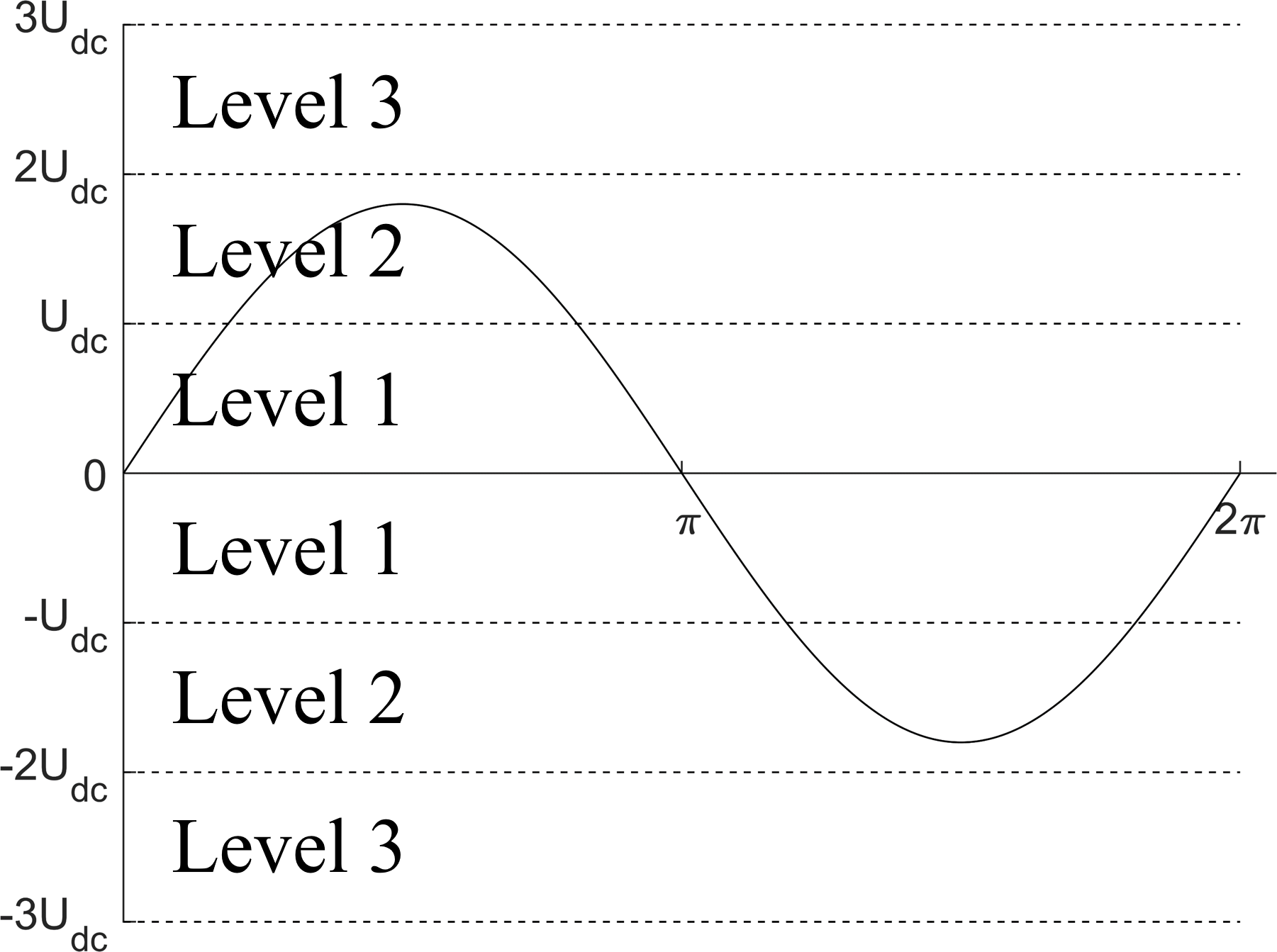}
\subcaption{Voltage level division of the modulation wave}
\end{subfigure}
\begin{subfigure}{0.48\textwidth}
\includegraphics[width=\textwidth]{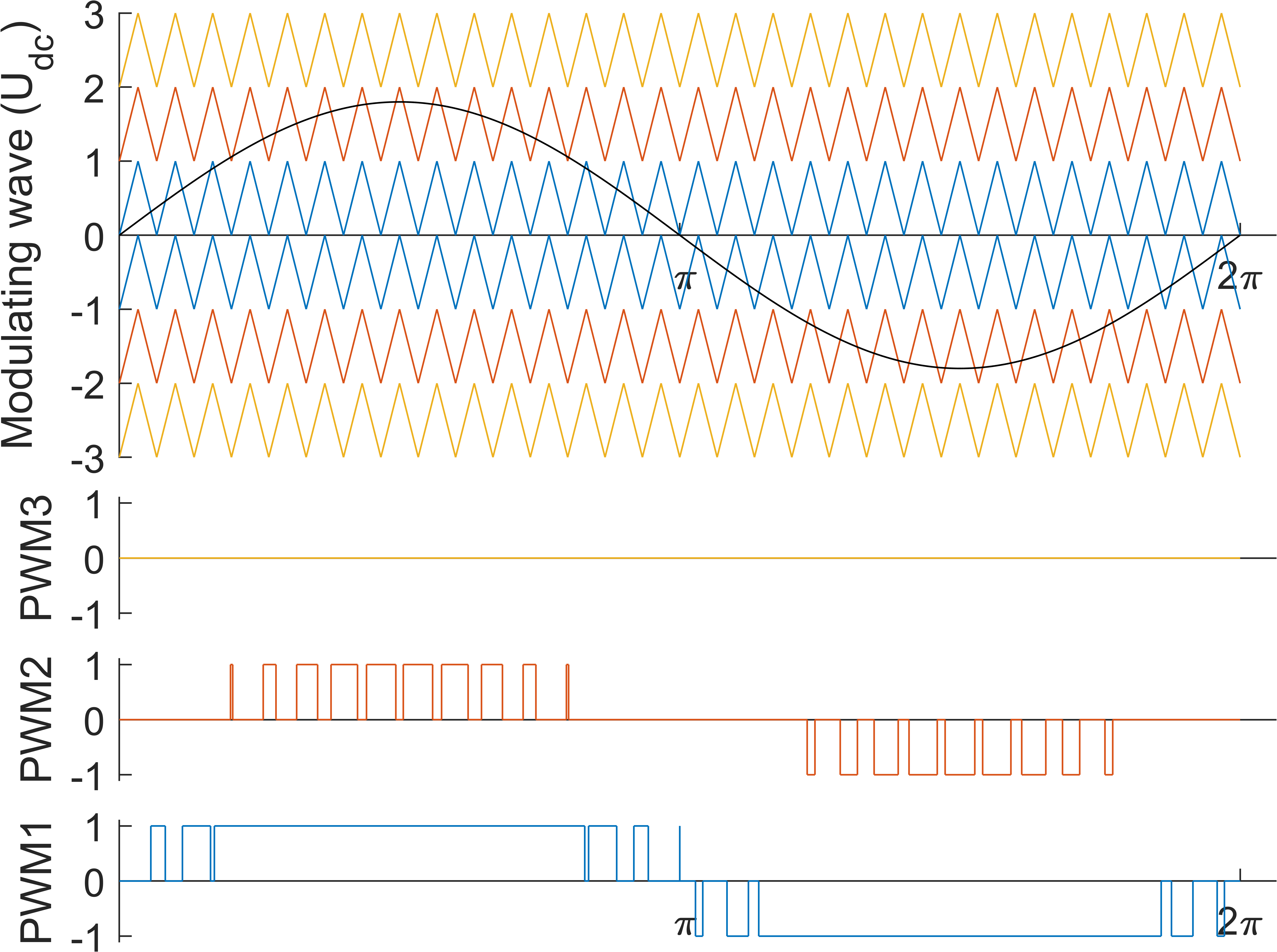}
\subcaption{PWM wave generation}
\end{subfigure}
\caption{The basic idea of level-shifted pulse width modulation control}
\label{LSPWM_fig}
\end{figure}

As shown in Fig.~\ref{LSPWM_fig}(a), the modulation wave (desired phase voltage reference) is quantized into several discrete \emph{voltage levels}. Each cell (H-bridge) is assigned to a voltage level, and the quantization step is approximately the cell/module DC voltage. For illustration, we assume all cell voltages are equal to \(U_{dc}\), although in practice they may differ slightly.

After the level assignment, each H-bridge generates a PWM waveform that tracks the duty cycle associated with its assigned level, as shown in Fig.~\ref{LSPWM_fig}(b). For a sinusoidal modulation wave, cells assigned to higher voltage levels are activated less frequently than those assigned to lower levels. In the particular example of Fig.~\ref{LSPWM_fig}(b), the cell assigned to Level~1 is used for most of the electrical cycle, whereas the cell assigned to Level~3 remains bypassed because the reference does not enter that band. The level assignment can be updated frequently, enabling either nearly equal long-term usage (if desired) or intentionally unbalanced usage to support SOH balancing.

It is also worth noting that although the output waveform (and the input waveform during AC charging) is a three-phase sinusoid, the modulation reference in Fig.~\ref{LSPWM_fig}(a) need not be purely sinusoidal. For example, third-harmonic injection and SVPWM-equivalent modulation can be combined with LSPWM to increase the achievable fundamental voltage \cite{SVPWM}. These methods shape the phase voltage while preserving the line-to-line voltage waveform. In this work, however, we focus on sinusoidal modulation for clarity; the proposed strategy can be extended to non-sinusoidal references in a straightforward manner.

Another scenario that should be mentioned is DC charging. For DC charging, the three branches of cells in Figure \ref{comparison_topology}(b) can either be connected in parallel or in series, depending on the charger's voltage and current limit. Note that this flexibility is another advantage of the cell-level inverter topology. Whichever connection is used, the control principle is similar to Figure \ref{LSPWM_fig}, except that the modulation wave is a DC voltage with a slowly changing magnitude. In this case, almost all the voltage levels' duty cycles in Figure  \ref{LSPWM_fig}(b) will become either one or zero, and only one voltage level can have a non-binary duty cycle. Specifically, for the example shown in Figure \ref{LSPWM_fig}, if the DC voltage is between $U_{dc}$ and $2U_{dc}$, only voltage level 2 will have a non-binary duty cycle. 

\subsection{Imbalanced operation within each phase and related theoretical analysis}\label{section_greedy}

While prior literature also noted LSPWM's performance under an imbalanced power distribution \cite{LSPWM, LSPWM2}, these works typically treated imbalance as a non-ideal condition and rarely deliberately exploited it. Here, we embrace and leverage LSPWM's ability to distribute throughput unevenly across cells. This raises a natural question: what is the maximum imbalance that LSPWM can realize? Equivalently, given the initial cell states, which final states are achievable under LSPWM?

Indeed, although LSPWM enables different cells to be charged at different rates, the achievable imbalance is still constrained by the duty-cycle structure of the modulation reference. In the illustrative case of Fig.~\ref{LSPWM_fig}, the first two voltage levels have nonzero duty cycles, meaning that at least two H-bridges must be active during each modulation period; thus, it is impossible to charge only one cell while charging none of the others. On the other hand, it is possible to charge Cells~1 and~2 equally while bypassing Cell~3 on average. For example, assign Cell~1 to Level~1, Cell~2 to Level~2, and Cell~3 to Level~3 during one modulation period, and swap the assignments of Cells~1 and~2 in the next period. Repeating this swapping yields equal average throughput for Cells~1 and~2 while Cell~3 remains bypassed.
In the general case, Theorem~\ref{t1} provides a necessary and sufficient condition relating the initial and final states. We first introduce the following definitions.

Let \(Q_{\text{initial},i}\) and \(Q_{\text{final},i}\) denote the initial charge capacity and the desired final charge capacity of cell \(i\), respectively, and define \(\Delta Q_i = Q_{\text{final},i}-Q_{\text{initial},i}\).
Let \(d_\ell\) denote the duty cycle of voltage level \(\ell\) in Fig.~\ref{LSPWM_fig}(b). We assume that the controller may reassign cells to voltage levels from one modulation period to the next.

\begin{definition}[Achievable]
A target final state \(\{Q_{\text{final},i}\}_{i=1}^{n_1}\) is \emph{achievable} from an initial state \(\{Q_{\text{initial},i}\}_{i=1}^{n_1}\) if there exist a time horizon \(\Delta t>0\) and an admissible level-assignment policy such that, under LSPWM operation, all cells reach \(Q_{\text{final},i}\) after \(\Delta t\).
\end{definition}

\begin{theorem}\label{t1}
Assume that (i) all cell DC voltages are equal and constant over time, so the duty-cycle pattern \(\{d_\ell\}_{\ell=1}^{n_1}\) induced by the modulation reference is constant, and (ii) the modulation period \(T_m\) is much shorter than the total charging time, i.e., \(T_m\ll \Delta t\), so that level assignments can be permuted many times over \(\Delta t\).
Let \(\Delta Q_{(1)}\ge \Delta Q_{(2)}\ge \cdots \ge \Delta Q_{(n_1)}\) be the sorted capacity gains across cells, and let \(d_{1}\ge d_{2}\ge \cdots \ge d_{n_1}\) denote the duty cycles across voltage levels (largest to smallest).
Then the target final state is achievable if and only if
\begin{equation}\label{condition}
\frac{\sum_{j=1}^{k}\Delta Q_{(j)}}{\sum_{j=1}^{n_1}\Delta Q_{(j)}}
\;\le\;
\frac{\sum_{j=1}^{k}d_{j}}{\sum_{j=1}^{n_1}d_{j}},
\qquad k=1,2,\ldots,n_1,
\end{equation}
(with equality at \(k=n_1\)).
Moreover, in LSPWM, \(d_i\) equals the fraction of time the reference exceeds the \(i\)th carrier threshold (level boundary); hence \(d_{1}\ge d_{2}\ge \cdots \ge d_{n_1}\) for any reference waveform.
\end{theorem}

\begin{proof}
We first prove necessity of \eqref{condition}. Fix any \(k\in\{1,\ldots,n_1\}\) and consider any set of \(k\) cells. Under Assumption~(i), over a fixed charging horizon \(\Delta t\), the total added charge (Ah) delivered by voltage level \(\ell\) is proportional to its duty cycle \(d_\ell\); hence the maximum fraction of total added charge that any \(k\) cells can receive is achieved when they are assigned to the \(k\) highest-duty levels \(\{1,\ldots,k\}\) for the entire charging duration. Therefore, for any feasible strategy, the fraction of total added charge received by the \(k\) most-charged cells cannot exceed \(\frac{\sum_{\ell=1}^k d_\ell}{\sum_{\ell=1}^{n_1} d_\ell}\). Since \(\Delta Q_{(1)}\ge\cdots\ge \Delta Q_{(n_1)}\) are the ordered capacity gains, this yields \eqref{condition}. Thus \eqref{condition} is necessary.

We next prove sufficiency by constructing a specific level-assignment policy and showing (by contradiction) that it achieves the target whenever \eqref{condition} holds. Consider the following level-assignment policy during charging. Let \(\Delta Q_i^r(t)\) denote the remaining charge needed by cell \(i\) at time \(t\), i.e., \(\Delta Q_i^r(t)=Q_{\text{final},i}-Q_i(t)\). At each modulation period, sort the cells by \(\Delta Q_i^r(t)\) in nonincreasing order and assign the cell with the \(j\)th largest \(\Delta Q_i^r(t)\) to voltage level \(j\). (If ties occur, preserve ties by cyclically rotating level assignments among tied cells.)

Under this policy, the ordering of the remaining demands is invariant: if \(\Delta Q_i^r(t_0)\ge \Delta Q_j^r(t_0)\) at some time \(t_0\), then \(\Delta Q_i^r(t)\ge \Delta Q_j^r(t)\) for all \(t\ge t_0\). To see this, suppose for contradiction that there exists a first time \(t^\star\) such that \(\Delta Q_i^r(t^\star)<\Delta Q_j^r(t^\star)\). By continuity and minimality of \(t^\star\), we must have \(\Delta Q_i^r(t^\star)=\Delta Q_j^r(t^\star)\), and immediately before \(t^\star\), \(\Delta Q_i^r>\Delta Q_j^r\). But then the policy assigns cell \(i\) to a level with duty no smaller than that of cell \(j\) (since higher remaining demand gets a higher-duty level), so cell \(i\) cannot decrease its remaining demand faster than cell \(j\) at the crossing instant, contradicting the existence of such a first crossing.

Now assume \eqref{condition} holds but, at the end of charging, not all targets are met; i.e., there exists at least one cell with \(\Delta Q_i^r(\Delta t)>0\). Let \(k\in\{1,\ldots,n_1-1\}\) be the number of cells with strictly positive remaining demand at time \(\Delta t\). By the ordering invariance above, these \(k\) cells must have been the top-\(k\) cells in remaining demand throughout the charging process, and therefore they were always assigned to the \(k\) highest-duty levels \(\{1,\ldots,k\}\). Consequently, the maximum fraction of total delivered charge available to these \(k\) cells is \(\frac{\sum_{\ell=1}^k d_\ell}{\sum_{\ell=1}^{n_1} d_\ell}\). However, since these \(k\) cells still require additional charge at the end, the required fraction of total charge for the top-\(k\) capacity gains must satisfy
\[
\frac{\sum_{j=1}^k \Delta Q_{(j)}}{\sum_{j=1}^{n_1}\Delta Q_{(j)}} >
\frac{\sum_{\ell=1}^k d_\ell}{\sum_{\ell=1}^{n_1} d_\ell},
\]
which contradicts \eqref{condition}. Therefore, all cells must satisfy \(\Delta Q_i^r(\Delta t)=0\), and the desired final state is achievable. This proves sufficiency.
\end{proof}

\begin{remark}
Condition \eqref{condition} characterizes achievable final states only when the duty cycles \(\{d_\ell\}\) are constant over time. In practice, the duty cycles depend on the cells' terminal voltage and the amplitude of the phase voltage, both of which can vary as the charging progresses. Nevertheless, while Theorem~\ref{t1} does not guarantee achievability under general voltage variations, the charging strategy we mentioned in the sufficiency proof of Theorem \ref{t1} is still very useful, as we explain below. 
\end{remark}

\begin{strategy}[remaining capacity gain balancing]\label{greedy1}
    At each modulation period, assigns the cell with the \(j\)th largest \(\Delta Q_i^r(t)\) to voltage level \(j\). If ties occur, preserve ties by cyclically rotating level assignments among tied cells.
\end{strategy}

This charging strategy can be interpreted as a greedy allocation rule: at each modulation period it assigns the highest-throughput voltage levels to the cells with the largest remaining required capacity gains \(\Delta Q_i^r(t)\), which tends to reduce the dispersion of \(\{\Delta Q_i^r(t)\}\) while driving all cells toward their targets. When cell terminal voltages are equal and approximately constant, the duty-cycle pattern \(\{d_\ell(t)\}\) depends only on the phase-voltage reference and is independent of the cells' state trajectory. In this case, Strategy~\ref{greedy1} is optimal in the following sense: if the target final state is achievable, the strategy attains it (as shown in the sufficiency proof of Theorem~\ref{t1}); if the target is not achievable, the strategy produces a final state that is as close as possible to the target within the limits imposed by the duty-cycle shares (i.e., it does not ``waste'' high-duty levels on cells with smaller remaining demand).

\begin{strategy}[remaining-capacity balancing]\label{greedy2}
At each modulation period, if the battery pack is discharging, assign the cell with the \(j\)th largest \(\Delta Q_i^r(t)\) to voltage level \(j\); if the battery pack is charging (e.g., in regenerative mode), assign the cell with the \(j\)th smallest \(\Delta Q_i^r(t)\) to voltage level \(j\). If ties occur, maintain ties by cyclically rotating level assignments among tied cells.
\end{strategy}

\begin{remark}
Strategy~\ref{greedy2} can be regarded as a special case of Strategy~\ref{greedy1} in which the desired final state is complete depletion of the usable pack capacity. Under the constant-terminal-voltage approximation, the strategy allocates higher Ah throughput to cells with larger remaining usable capacity during discharge, thereby minimizing residual unusable capacity at pack depletion and maximizing extractable energy (and thus range). In the general case, it remains an effective and computationally light heuristic, consistent with the intuition that lower-capacity cells should be discharged more slowly to avoid prematurely limiting the pack.
\end{remark}

Strategy~\ref{greedy2} is essential for operating a cell-level inverter because it provides an effective and computationally efficient way to increase usable discharge capacity (and thus driving range). It is closely related to SOC balancing, which is presented below. 

\begin{strategy}[SOC balancing]\label{greedy3}
At each modulation period, if the battery pack is discharging, assign the cell with the \(j\)th largest SOC to voltage level \(j\); if the battery pack is charging (e.g., in regenerative mode), assign the cell with the \(j\)th smallest SOC to voltage level \(j\). If ties occur, maintain ties by cyclically rotating level assignments among tied cells.
\end{strategy}

Strategies \ref{greedy2} and \ref{greedy3} coincide when cells have identical SOH, but they differ when cells exhibit SOH heterogeneity: SOC balancing equalizes SOC, whereas remaining-capacity balancing accounts for differences in usable capacities and can therefore extract more energy under heterogeneous SOH. In practice, SOC balancing is often easier to implement because SOC can be estimated without explicitly estimating SOH (or capacity), whereas remaining-capacity balancing is most accurate when both SOC and SOH are available. Since most existing works do not target SOH balancing, SOC balancing is more commonly adopted than remaining-capacity balancing.

%In real-life driving, the battery current can be positive when the vehicle decelerates. To preserve the optimality described in Corollary \ref{l1}, the voltage level assignment must be reversed when the current is positive. In other words, the cell with the $i^{th}$ lowest remaining capacity should be assigned to the $i^{th}$ voltage level to reduce the difference in remaining capacity during charging. Nevertheless, the driver's behavior is usually unpredictable during driving, and the controller cannot know when the vehicle will decelerate in advance. Therefore, such an extension is not implemented in this paper.

\subsection{Power distribution among the three phases}

For the three-phase H-bridge-based cell-level inverter in Fig.~\ref{comparison_topology}(b), power imbalance can be introduced not only among cells within a phase but also across the three phases. For a given operating point (i.e., a commanded set of motor voltages), the line-to-line voltages \(V_{AB}(t),V_{BC}(t),V_{CA}(t)\) are balanced sinusoids with the same frequency and amplitude and phase shifts of \(120^\circ\). Their phasors \(\dot U_{AB},\dot U_{BC},\dot U_{CA}\) therefore form an equilateral triangle, as illustrated in Fig.~\ref{imbalance}(a). While these line-to-line phasors are fixed once the motor command is specified, the corresponding phase-voltage phasors \(\dot U_{OA},\dot U_{OB},\dot U_{OC}\) are not unique: the neutral (common-mode) voltage can be shifted, which changes the phase-voltage magnitudes without altering the line-to-line voltages (Fig.~\ref{imbalance}(b)). This degree of freedom enables different phases to process different amounts of power, potentially allowing cells in different phases to be discharged at different rates.

\begin{figure}[htbp]
\centering
\includegraphics[scale=0.4]{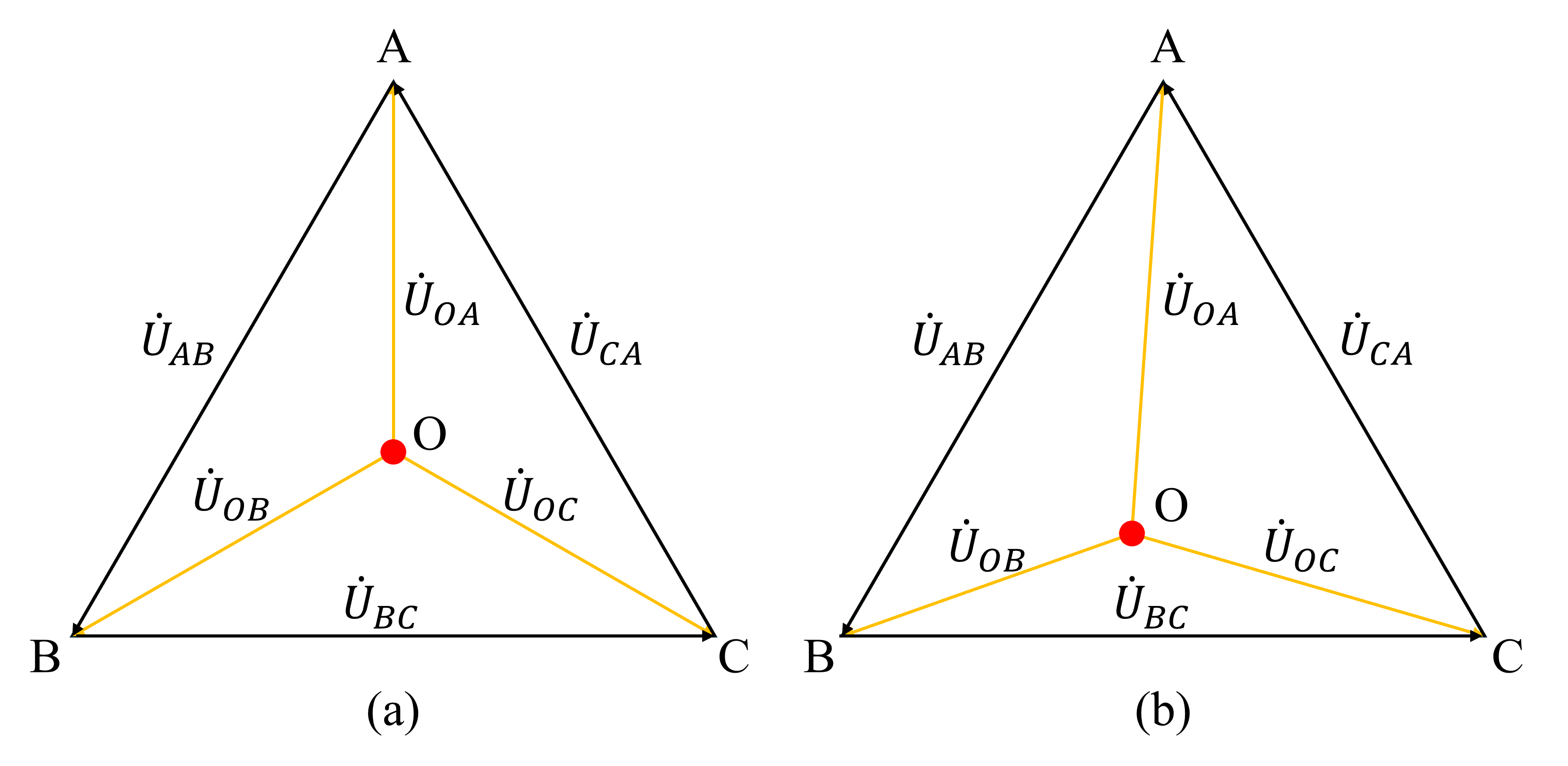}
\caption{(a) Balanced and (b) imbalanced power distributions among the three phases of the H-bridge-based cell-level inverter topology}
\label{imbalance}
\end{figure}

Although the topology permits imbalance both within and across phases, in this work we intentionally introduce imbalance only {within} each phase. The reason is that within-phase imbalance can be realized through LSPWM level assignments without changing the commanded phase voltage, and therefore does not inherently increase motor-side copper losses. In contrast, deliberate imbalance across phases increases conduction losses (and thus heat generation) for the same required power. This can be illustrated by a simple example on the source side: assume the RMS line-to-line voltage magnitude is \(\sqrt{3}U_0\) and each phase has an effective internal resistance \(R\) (including cell and power-electronics conduction resistances). Under balanced operation (Fig.~\ref{imbalance}(a)), each phase has RMS voltage \(U_0\) and RMS current \(I\), giving total conduction loss \(3I^2R\). In an extreme imbalanced case where point \(O\) overlaps point \(C\) in Fig.~\ref{imbalance}(b), the phase RMS voltages become \(\sqrt{3}U_0,\sqrt{3}U_0,0\), and the phase RMS currents become \(\sqrt{3}I,\sqrt{3}I,0\), yielding total loss \(2(\sqrt{3}I)^2R=6I^2R\), i.e., twice the balanced loss. Therefore, we do not intentionally imbalance power among phases in this paper.

\section{Control strategy}

\subsection{Basic Control Idea}
\begin{figure}[htbp]
\centering
\includegraphics[scale=0.7]{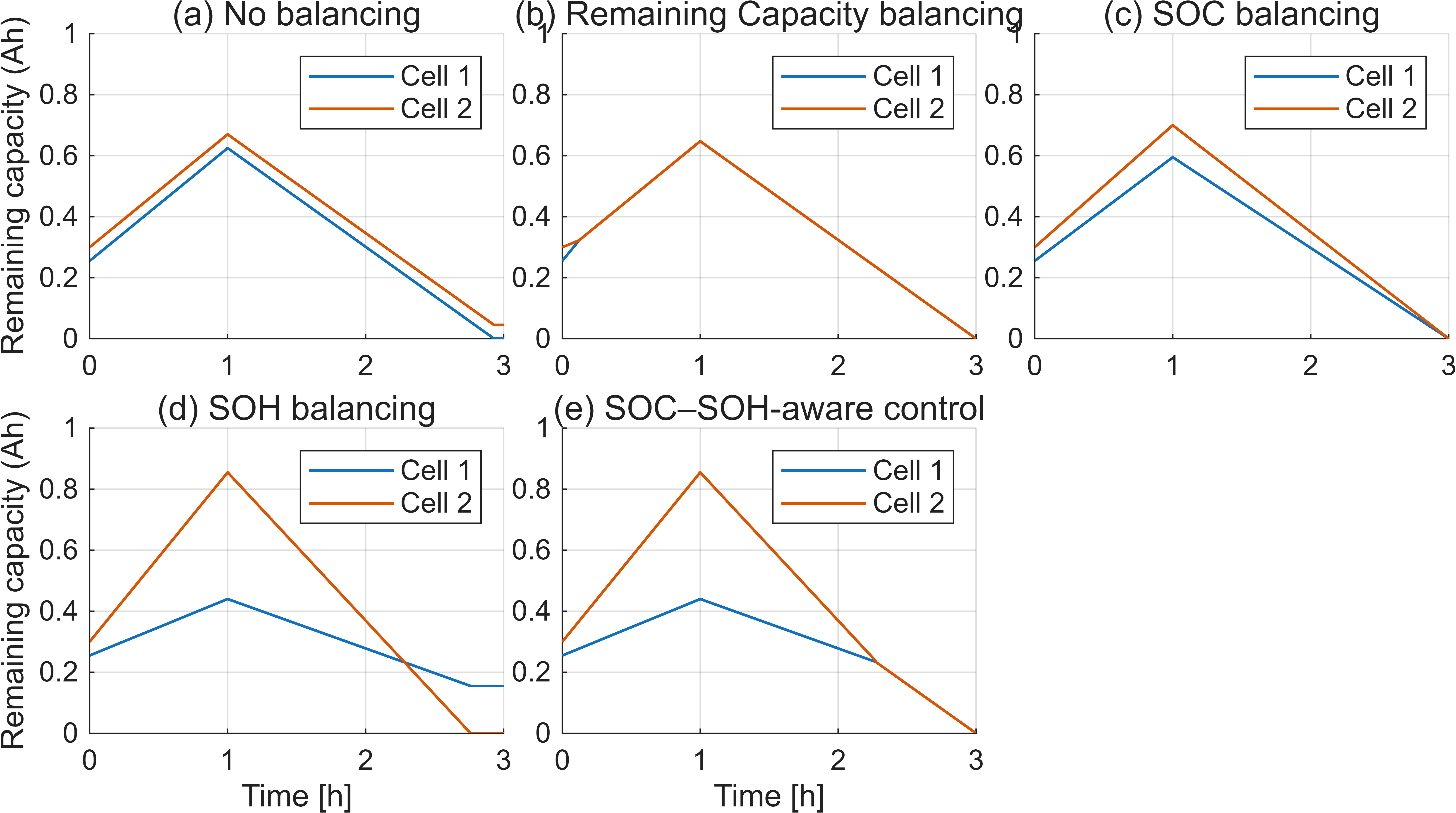}
\caption{Comparison of different control methods for the cell-level inverter with two cells in total. The battery pack is first charged from 40\,\% SOC to 70\,\% SOC and is then fully discharged.}
\label{idea}
\end{figure}

The idea of jointly optimizing battery lifetime and energy utilization can be illustrated using a simple example. In Fig.~\ref{idea}, we consider a pack with two cells. Cell~2 has a higher SOH than Cell~1, and both cells start at 30\,\% SOC with nominal capacity 1~Ah. The pack is first charged with a fixed current profile (i.e., fixed AC current magnitude on the source side) to a target pack SOC of 70\,\%, and is then fully discharged. Although the overall current profile is the same in Fig.~\ref{idea}(a)--(e), different strategies distribute the Ah throughput differently across the two cells. If the two cells are charged and discharged equally, the remaining capacities follow Fig.~\ref{idea}(a). Because Cell~1 has lower SOH (and thus lower usable capacity), it is depleted first during discharge, leaving a portion of Cell~2's capacity unused. If the strategy instead performs remaining-capacity balancing (Strategy \ref{greedy2} in Section \ref{section_greedy}), the two cells maintain equal remaining capacity after an initial balancing step, and the usable energy can be fully extracted, as shown in Fig.~\ref{idea}(b). However, since the Ah throughput of the two cells remains nearly identical, Cell~1 will still tend to reach EOL earlier over repeated cycling, wasting the remaining life of Cell~2. A similar conclusion applies to SOC balancing (Strategy \ref{greedy3} in Section \ref{section_greedy}) in Fig.~\ref{idea}(c): although SOC balancing can improve energy utilization (enabled by the ability to bypass cells in the cell-level inverter topology), it does not preferentially reduce the stress on the weaker cell, so SOH disparity can grow over time and pack lifetime is not optimized.

In contrast, under SOH balancing (Fig.~\ref{idea}(d)), Cell~2 is deliberately used more than Cell~1, which reduces the SOH gap and can cause the two cells to reach EOL more simultaneously, minimizing wasted lifetime. The trade-off is that Cell~2 is depleted faster during discharge, which can reduce usable energy and hence driving range. Fig.~\ref{idea}(e) illustrates a more favorable strategy: it resembles SOH balancing during charging by allocating most Ah throughput to Cell~2, but it also performs remaining-capacity balancing during discharge so that no energy is left unused. This combined SOH + capacity balancing can therefore improve both long-term lifetime and energy utilization.

Finally, note that in Fig.~\ref{idea}(e) we apply SOH balancing during charging and remaining-capacity balancing during discharging, rather than the reverse. The primary reason is that battery degradation is SOC-dependent: high-SOC operation during charging can accelerate aging via both calendar aging mechanisms and charge-induced stress \cite{SOCreview}. In contrast, although discharging contributes to cyclic aging, it also reduces SOC and can slow calendar aging, so its net impact on degradation is more uncertain \cite{SOCreview2}. Therefore, if SOC balancing is performed during charging while SOH balancing is attempted during discharge, the correct discharge allocation is ambiguous: it is unclear whether the weakest cells should be discharged more (to reduce time spent at high SOC) or less (to reduce cyclic-aging stress).
Moreover, traction discharge power is exogenous and uncertain compared with scheduled charging, so implementing SOH balancing during discharge is considerably more complex; we therefore focus on SOH balancing during charging in this paper.

\subsection{The proposed charging strategy}
As discussed earlier, the core idea of ``SOC–SOH-aware'' control is to perform SOH balancing during charging and remaining-capacity balancing during discharge. The discharge strategy has been presented in Strategy~\ref{greedy2}. In this section, we focus on the charging controller, whose goal is to promote SOH balancing without creating an SOC distribution that cannot be rebalanced during the subsequent discharge.

\subsubsection{Model simplifications}
During charging, the controllable variables include the phase voltage, line current, and the cells' voltage-level assignments. These variables---especially the level assignments---can be updated at each control interval. Without simplification, the number of decision variables grows with the charging horizon and quickly becomes very large even for a small number of cells. Moreover, the level assignments are discrete (integer) decisions, which turn the problem into a large-scale mixed-integer optimization that is generally computationally intractable.

The primary purpose of the model simplifications is to reduce the degrees of freedom so that the number of decision variables does not grow with charging time. In Strategy~\ref{greedy1}, we showed that an assignment sequence can be generated once the desired final cell states are specified. Therefore, rather than optimizing the integer assignments at every time step, we optimize only the desired final states. For the phase voltage, we adopt a constant-voltage approximation and operate at the maximum permissible voltage magnitude (subject to cell and hardware voltage limits). As illustrated in Fig.~\ref{LSPWM_fig}, a higher phase-voltage magnitude corresponds to more cells connected in series, which allows a smaller line current for a fixed charging power. Since degradation often increases with C-rate, operating at the highest admissible phase voltage can reduce aging for a given charging-time requirement.

Finally, we assume a CCCV charging profile, which is widely used for EV batteries \cite{cccv}. Specifically, cells are charged at a constant current until the maximum cell voltage reaches its limit; afterward, charging proceeds in a constant-voltage stage. Because the constant-voltage stage is largely determined by the cell states and voltage limit, the main continuous degree of freedom is the constant current in the CC stage.

With these simplifications, the remaining decision variables are:
\begin{itemize}
    \item the target final added charge (capacity gain) of each cell, and
    \item the line current during the constant-current stage.
\end{itemize}
Although the number of decision variables is now small, the problem remains challenging because the charging-time constraint is difficult to express analytically. The total charging time depends on both the current profile and the required cell-wise capacity gains, and the current is not constant throughout the CCCV process. To further simplify the time constraint, we discretize the CCCV profile into multiple stages: Stage~0 represents constant-current charging, and Stages~1 through \(m\) represent constant-voltage charging. An example of this stage division is shown in Fig.~\ref{simpli}(a).

\begin{figure}[htbp]
\centering
\includegraphics[scale=0.7]{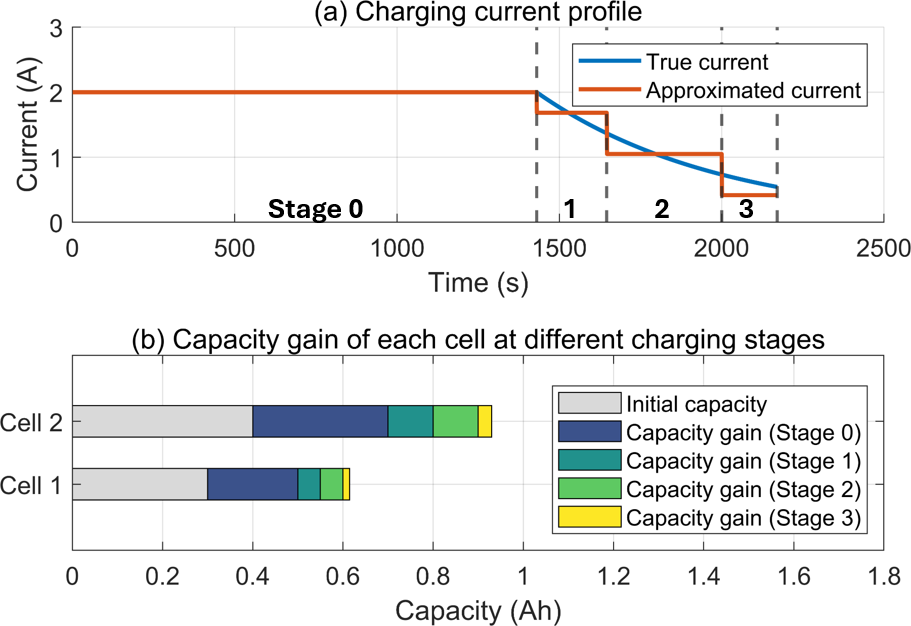}
\caption{The division of the charging stages.}
\label{simpli}
\end{figure}

In the example shown in Fig.~\ref{simpli}, the CCCV charging is divided into four stages. The stage division is based on SOC, and the upper SOC boundary of Stage \(i\) is denoted by \(SOC_{\max,i}\). We assume that the maximum admissible (normalized) C-rate is a function of SOC, denoted by \(g(SOC)\). Stage~0 represents constant-current charging with a fixed line-current magnitude \(I_{cc}\), which corresponds to a normalized C-rate \(I_{cc}/Q_{no}\). In the current-envelope-limited setting, Stage~0 ends when this commanded C-rate reaches the admissible envelope, i.e.,
\begin{equation}
SOC_{\max,0}=g^{-1}\!\left(\frac{I_{cc}}{Q_{no}}\right).
\end{equation}
The remaining SOC interval \([SOC_{\max,0},1]\) is then approximated by \(m\) constant-voltage sub-stages (Stages \(1\) through \(m\)).

Specifically, we choose the stage boundaries to uniformly partition \([SOC_{\max,0},1]\) in SOC:
\begin{equation}\label{maxSOC}
SOC_{\max,i}=
\begin{cases}
g^{-1}\!\left(\dfrac{I_{cc}}{Q_{no}}\right), & i=0,\\[2mm]
SOC_{\max,0}+\dfrac{i}{m}\left(1-SOC_{\max,0}\right), & i=1,2,\ldots,m.
\end{cases}
\end{equation}
For each stage, we approximate the C-rate magnitude by a constant \(I^{no}_{\mathrm{avg},i}\) defined as
\begin{equation}\label{Iavg}
I^{no}_{\mathrm{avg},i}=
\begin{cases}
\dfrac{I_{cc}}{Q_{no}}, & i=0,\\[2mm]
\dfrac{g(SOC_{\max,i-1})+g(SOC_{\max,i})}{2}, & i=1,2,\ldots,m,
\end{cases}
\end{equation}
where \(Q_{no}\) is the nominal cell capacity.

Within each stage, we assume the current magnitude (and thus the C-rate) is constant and equals \(I^{no}_{\mathrm{avg},i}\) (unit: C or 1/h). Under this approximation, the capacity gain delivered in each stage is proportional to that stage's duration. Let \(Q_{p,i}\) denote the capacity gain (Ah) of cell \(p\) in Stage \(i\). Then the total charging time can be expressed as a linear function of \(\{Q_{p,i}\}\). By using \(\{Q_{p,i}\}\) as decision variables, the time constraint becomes linear, substantially simplifying the optimization. This approximation becomes increasingly accurate as the number of stages \(m\) increases, at the cost of additional decision variables and computation time.

\subsubsection{Formulation of the optimization problem}
In summary, the model simplifications in the previous subsection are:
\begin{itemize}
    \item The phase-voltage magnitude \(U_{\text{phase}}\) is set to the highest admissible constant (subject to voltage limits).
    \item The charging profile is approximated by CCCV charging, discretized into \(m{+}1\) stages, with a constant line-current magnitude in each stage.
    \item The decision variables are the per-cell capacity gains in each stage, denoted by \(Q_{i,j}\) for cell \(i\) and stage \(j\). Note that the CC current magnitude \(I_{cc}\) is treated as a fixed parameter in the formulation below, and is optimized separately via a grid search.
\end{itemize}
We next define an objective function and constraints in terms of these decision variables. Ideally, the cost would reflect the capacity loss directly; however, this would require an explicit aging model, thereby reducing generality. Instead, we adopt a surrogate objective that discourages charging low-SOH cells and high-C-rate operation. For fixed \(U_{\text{phase}}\) and \(I_{cc}\), we solve
\begin{equation}\label{cf}
    \min_{\{Q_{i,j}\}} \sum_{i=1}^{n_1}\sum_{j=0}^{m} w_{i,j}\,Q_{i,j},
\end{equation}
where the weights are defined as
\begin{equation}\label{weight}
w_{i,j}=
\begin{cases}
\dfrac{1+\kappa I^{no}_{\mathrm{avg},j}}{(SOH_i-SOH_{\mathrm{EOL}})^2}, & SOH_i>SOH_{\mathrm{EOL}}+\varepsilon,\\[2mm]
\left(1+\kappa I^{no}_{\mathrm{avg},j}\right) M, & SOH_i\le SOH_{\mathrm{EOL}}+\varepsilon,
\end{cases}
\end{equation}
with \(\kappa=0.1\), \(\varepsilon=10^{-3}\), and a large constant \(M=10^{6}\).
Here \(SOH_i\) is the SOH of cell \(i\), and \(SOH_{\mathrm{EOL}}\) is the EOL threshold (70\,\% by default). This choice of \(w_{i,j}\) penalizes allocating charge to low-SOH cells and to stages with higher C-rate; once a cell reaches EOL, the large weight effectively prevents further charging, so the cell is bypassed during charging.

The constraints of the optimization problem are listed below:
\begin{enumerate}
\item \textbf{Nonnegativity of stage-wise allocations.}
\begin{equation}
Q_{i,j}\ge 0,\qquad \forall i,\forall j.
\end{equation}

\item \textbf{Fixed total added charge content.}
\begin{equation}
\sum_{i=1}^{n_1}\sum_{j=0}^{m} Q_{i,j}
=
Q_{\mathrm{final,sum}}-Q_{\mathrm{initial,sum}},
\end{equation}
where $Q_{\mathrm{final,sum}}$ and $Q_{\mathrm{initial,sum}}$ denote the target and initial \emph{total charge contents} (Ah) summed across all $n_1$ cells, respectively.

\item \textbf{Charging-time limit.}
\begin{equation}
\sum_{j=0}^{m}
\frac{\sum_{i=1}^{n_1} Q_{i,j}}
{Q_{no}\, I^{no}_{\mathrm{avg},j}\, \sum_{k=1}^{n_1} d_{k,j}}
\le t_{\mathrm{total}},
\end{equation}
where $t_{\mathrm{total}}$ is the charging-time limit (hours), and $d_{i,j}$ is the duty cycle of voltage level $i$ in stage $j$.
For sinusoidal modulation, we approximate
\begin{equation}\label{dij}
d_{i,j}=\frac{2}{\pi}\arccos\!\left(
\min\!\left(\frac{(2i-1)U_{\mathrm{ter},j}}{2U_{\mathrm{phase}}},\,1\right)
\right),
\end{equation}
where $U_{\mathrm{ter},j}$ is the representative per-cell terminal voltage in stage $j$ and $U_{\mathrm{phase}}$ is the (fixed) phase-voltage magnitude. For DC modulation, we use
\begin{equation}
d_{i,j}=
\begin{cases}
0, & (i-1)U_{\mathrm{ter},j}>U_{\mathrm{phase}},\\
1, & iU_{\mathrm{ter},j}<U_{\mathrm{phase}},\\
\dfrac{U_{\mathrm{phase}}-(i-1)U_{\mathrm{ter},j}}{U_{\mathrm{ter},j}}, & \text{otherwise}.
\end{cases}
\end{equation}
Strictly speaking, terminal voltages vary across cells and time. We approximate the stage-wise terminal voltage by the OCV at an average SOC plus an average ohmic drop:
\begin{equation}\label{ter_appro}
U_{\mathrm{ter},j}
\approx
f_{OCV}\!\left(\frac{Q_{\mathrm{final,sum}}+Q_{\mathrm{initial,sum}}}{2\sum_{i=1}^{n_1} Q_{\max,i}}\right)
+ Q_{no}\, I^{no}_{\mathrm{avg},j} R_0,
\end{equation}
where $f_{OCV}(\cdot)$ is the OCV--SOC map, $Q_{\max,i}$ is the current maximum capacity of cell $i$, and $R_0$ is an average source-side internal resistance (cell plus power-electronics conduction resistance).

\item \textbf{Per-stage SOC upper bounds (voltage-limit proxy).}
At the end of each stage $k$, the SOC of each cell must not exceed $SOC_{\max,k}$ defined in~\eqref{maxSOC}:
\begin{equation}
Q_{\mathrm{initial},i}+\sum_{j=0}^{k} Q_{i,j}
\le
Q_{\max,i}\, SOC_{\max,k},
\qquad \forall i,\forall k.
\end{equation}

\item \textbf{Stage-wise achievability under LSPWM.}
According to Theorem~\ref{t1}, for each stage $j$, the target allocation vector $\{Q_{i,j}\}_{i=1}^{n_1}$ must satisfy
\begin{equation}\label{strong}
\frac{\sum_{i=1}^{k} Q_{(i),j}}{\sum_{i=1}^{n_1} Q_{i,j}}
\le
\frac{\sum_{\ell=1}^{k} d_{\ell,j}}{\sum_{\ell=1}^{n_1} d_{\ell,j}},
\qquad \forall k\in\{1,\ldots,n_1\},\ \forall j\in\{0,\ldots,m\},
\end{equation}
where $Q_{(1),j}\ge \cdots \ge Q_{(n_1),j}$ denotes the sorted values of $\{Q_{i,j}\}_{i=1}^{n_1}$. Because level assignments can be permuted independently across stages and each stage is treated as a constant-duty allocation subproblem, imposing~\eqref{strong} for each stage is sufficient under the stage-wise approximation.

A direct implementation of~\eqref{strong} involves sorting. We instead enforce it using an epigraph formulation. Introduce auxiliary variables $t_{k,j}\in\mathbb{R}$ and $s_{i,k,j}\ge 0$ such that
\begin{equation}\label{topk_epigraph}
s_{i,k,j}\ge Q_{i,j}-t_{k,j},\qquad s_{i,k,j}\ge 0,\qquad \forall i,\forall k,\forall j.
\end{equation}
For any fixed $(k,j)$, minimizing $k\,t_{k,j}+\sum_i s_{i,k,j}$ over $t_{k,j}$ yields $\sum_{i=1}^k Q_{(i),j}$. Therefore,~\eqref{strong} is equivalently enforced by requiring the \emph{existence} of $(t_{k,j},s_{i,k,j})$ satisfying~\eqref{topk_epigraph} and
\begin{equation}\label{strong_lp}
k\,t_{k,j}+\sum_{i=1}^{n_1} s_{i,k,j}
\le
\left(\frac{\sum_{\ell=1}^{k} d_{\ell,j}}{\sum_{\ell=1}^{n_1} d_{\ell,j}}\right)
\sum_{i=1}^{n_1} Q_{i,j},
\qquad \forall k,\forall j.
\end{equation}

\item \textbf{SOH-consistent ordering of final stored charge (structural constraint).}
Cells with higher SOH are required to have higher stored charge by the end of charging:
\begin{equation}\label{higher}
Q_{\mathrm{initial},i}+\sum_{k=0}^{m} Q_{i,k}
\le
Q_{\mathrm{initial},j}+\sum_{k=0}^{m} Q_{j,k},
\qquad \text{if } SOH_i<SOH_j.
\end{equation}
This constraint aligns the solution with the intended SOH-aware usage pattern and facilitates the discharge-feasibility constraint below.

\item \textbf{Sufficient condition for full utilization during subsequent discharge.}
We consider the target discharge end state $Q_{\mathrm{end},i}=0$ for all cells; hence the required discharge amounts are $\Delta Q_i = Q_{\mathrm{final},i}$, where $Q_{\mathrm{final},i}\triangleq Q_{\mathrm{initial},i}+\sum_{k=0}^{m}Q_{i,k}$.
Applying Theorem~\ref{t1} yields the sufficient condition
\begin{equation}\label{enough}
\frac{\sum_{i=1}^k Q_{\mathrm{final},(i)}}{\sum_{i=1}^{n_1} Q_{\mathrm{final},i}}
\le
\frac{\sum_{i=1}^k d'_{i}}{\sum_{i=1}^{n_1} d'_{i}},
\qquad \forall k\in\{1,\ldots,n_1\},
\end{equation}
where $Q_{\mathrm{final},(1)}\ge\cdots\ge Q_{\mathrm{final},(n_1)}$ are the sorted final charge contents and $\{d'_i\}$ is a representative discharge duty-cycle pattern computed from average discharge voltage and phase-voltage magnitude. Similar to~\eqref{dij}, we approximate
\begin{equation}
d'_i=\frac{2}{\pi}\arccos\!\left(
\min\!\left(\frac{(2i-1)U_{\mathrm{dis,ter}}}{2U_{\mathrm{dis,phase}}},\,1\right)
\right),
\end{equation}
with
\begin{equation}\label{udis_ter_appro}
U_{\mathrm{dis,ter}}
\approx
f_{OCV}\!\left(\frac{Q_{\mathrm{final,sum}}}{2\sum_{i=1}^{n_1}Q_{\max,i}}\right)
- Q_{no}\, I^{no}_{\mathrm{dis,avg}}\,R_0,
\end{equation}
where $I^{no}_{\mathrm{dis,avg}}$ is the average C rate during discharge. For a more conservative design, $I^{no}_{\mathrm{dis,avg}}$ may be replaced by the maximum allowable discharge C-rate. Because~\eqref{higher} enforces a consistent ordering between SOH and the final stored charge, we may equivalently interpret “$(i)$” in~\eqref{enough} as the index of the cell with the $i$th highest SOH (using any consistent tie-breaking rule). Since the SOHs are known, \eqref{enough} can be enforced directly in the implementation.
\end{enumerate}

For any fixed $U_{\mathrm{phase}}$ and $I_{cc}$, the problem \eqref{cf} subject to Constraints~1--7 is a linear program (LP) in the decision variables $\{Q_{i,j}\}$ and can be solved efficiently. The overall procedure for selecting $U_{\mathrm{phase}}$, $I_{cc}$, and $\{Q_{i,j}\}$ is summarized in Fig.~\ref{flow}. We initialize $U_{\mathrm{phase}}$ and $I_{cc}$ at their maximum admissible values $(U_{\mathrm{phase,max}}, I_{cc,\max})$. A larger $U_{\mathrm{phase}}$ increases the minimum number of active voltage levels and therefore reduces the achievable imbalance; accordingly, we decrease $U_{\mathrm{phase}}$ until the LP becomes feasible. After $U_{\mathrm{phase}}$ is determined, we enumerate candidate values of $I_{cc}$ and select the pair $(I_{cc},\{Q_{i,j}\})$ that minimizes the objective in \eqref{cf}.

\begin{figure}[htbp]
\centering
\includegraphics[scale=0.6]{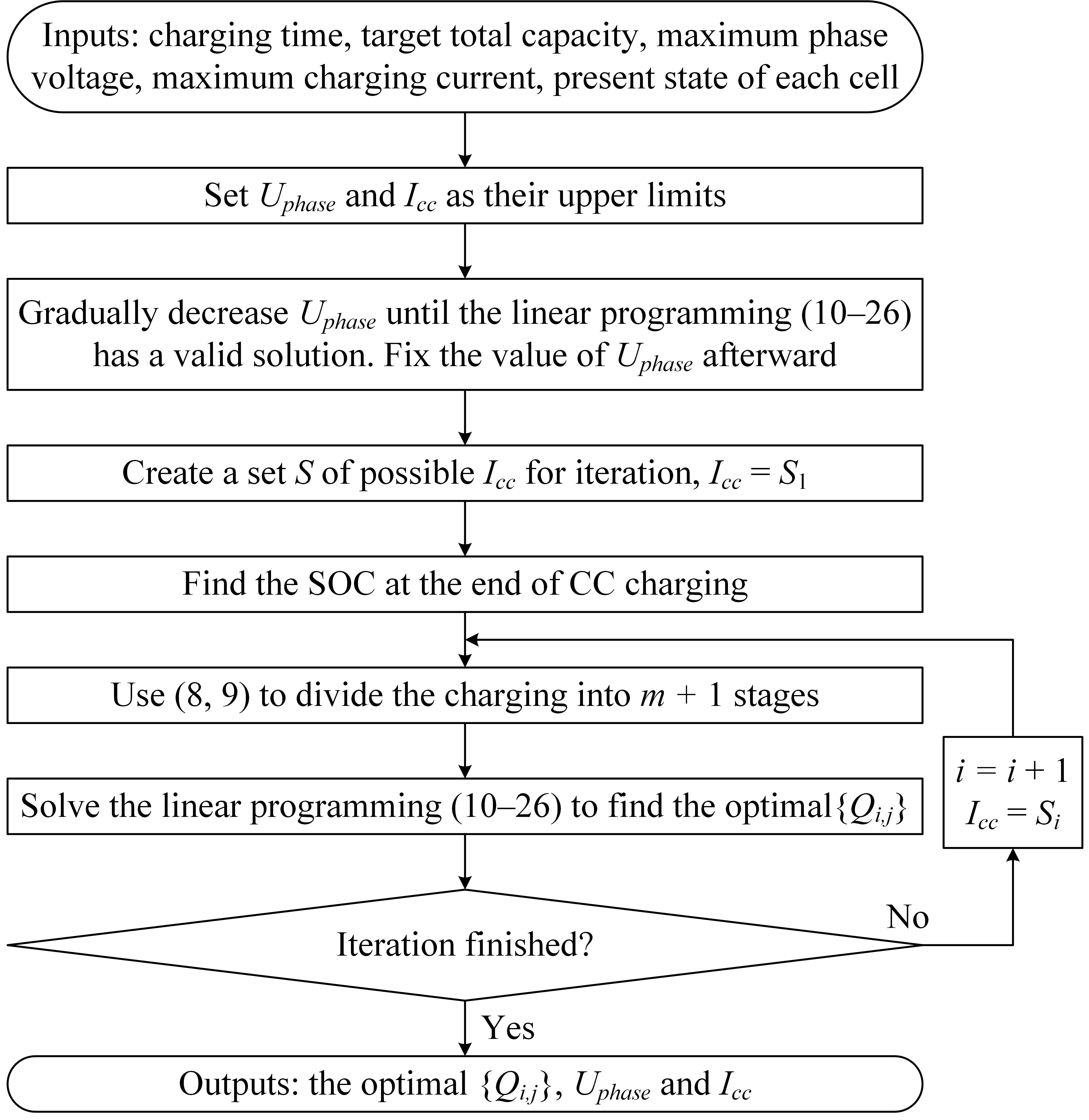}
\caption{The flow chart of the optimization algorithm.}
\label{flow}
\end{figure}

In some cases, the driver may want to charge the vehicle to a specific SOC level as quickly as possible. This can be achieved by skipping the step of finding the optimal $I_{cc}$ and setting $I_{cc}$ to the cell's maximum charging current. In other words, only the first two steps in Figure \ref{flow} are required, and the rest can be skipped.

After the optimization problem is solved, the phase voltage and line current are set as the optimized values. In implementation, the level assignments are adjusted online via Strategy~\ref{greedy1}, which provides robustness to deviations from the average-voltage approximations used in the LP.

\section{Simulation study}
\subsection{Simulation setup}\label{Section_aging}
%A paragraph that justifies no experiment
The lifetime of a battery pack is largely determined by the fastest-aging cells, and cell aging rates typically exhibit substantial heterogeneity. As a result, experimentally validating the proposed strategy can be expensive and time-consuming, as many cells and multiple tests are required to obtain statistically significant results. Therefore, as the first study to introduce the concept of ``SOC–SOH-aware control,'' we primarily demonstrate its effectiveness through comprehensive simulations under a wide range of scenarios and parameter settings. The code for this study is available at \href{https://github.com/Shida-Jiang/Model-Agnostic-Control-for-EV-Range-and-Lifetime-Extension}{https://github.com/Shida-Jiang/Model-Agnostic-Control-for-EV-Range-and-Lifetime-Extension}.

Specifically, we perform two types of simulations. First, a drive-cycle simulation examines how each cell's SOC evolves during charging and discharging. The discharge current profile is extracted from a real-world driving test and repeatedly used until the pack energy is depleted. The results (similar to Fig.~\ref{idea}(e)) illustrate how the proposed strategy intentionally creates non-uniform utilization during charging (i.e., different cells receive different effective charge throughput) while still enabling full energy utilization during discharge. 

Second, a long-term aging simulation quantifies the lifetime improvement of the proposed strategy relative to conventional SOC balancing. To construct realistic usage conditions, we use charging records collected from a 2021 Tesla Model~3 over approximately 2.5 years, comprising 255 cycles, as shown in Fig.~\ref{aging_exp}(a). In the simulation, these profiles are repeated until the battery pack reaches EOL. Among these charging events, 81 are DC fast charging with an average C rate greater than 0.2 C, and the remaining 174 are AC charging. Since the purpose of DC fast charging is typically to achieve sufficient capacity as quickly as possible, we always use the highest available voltage and current to maximize charging speed and suppress SOH balancing. Meanwhile, for each AC charging event, the estimated initial and final capacities and the charging duration are provided as inputs to the optimization algorithm in Fig.~\ref{flow}, whose output determines the cell-level charging commands. During discharge, we apply SOC balancing, and each discharge terminates when the pack's SOC reaches the recorded initial SOC of the subsequent charging session. This procedure is repeated until the pack reaches the EOL criterion defined in Section~\ref{end_pack}.

\begin{figure}[htbp]
\centering
\begin{subfigure}{0.43\textwidth}
\includegraphics[width=\textwidth]{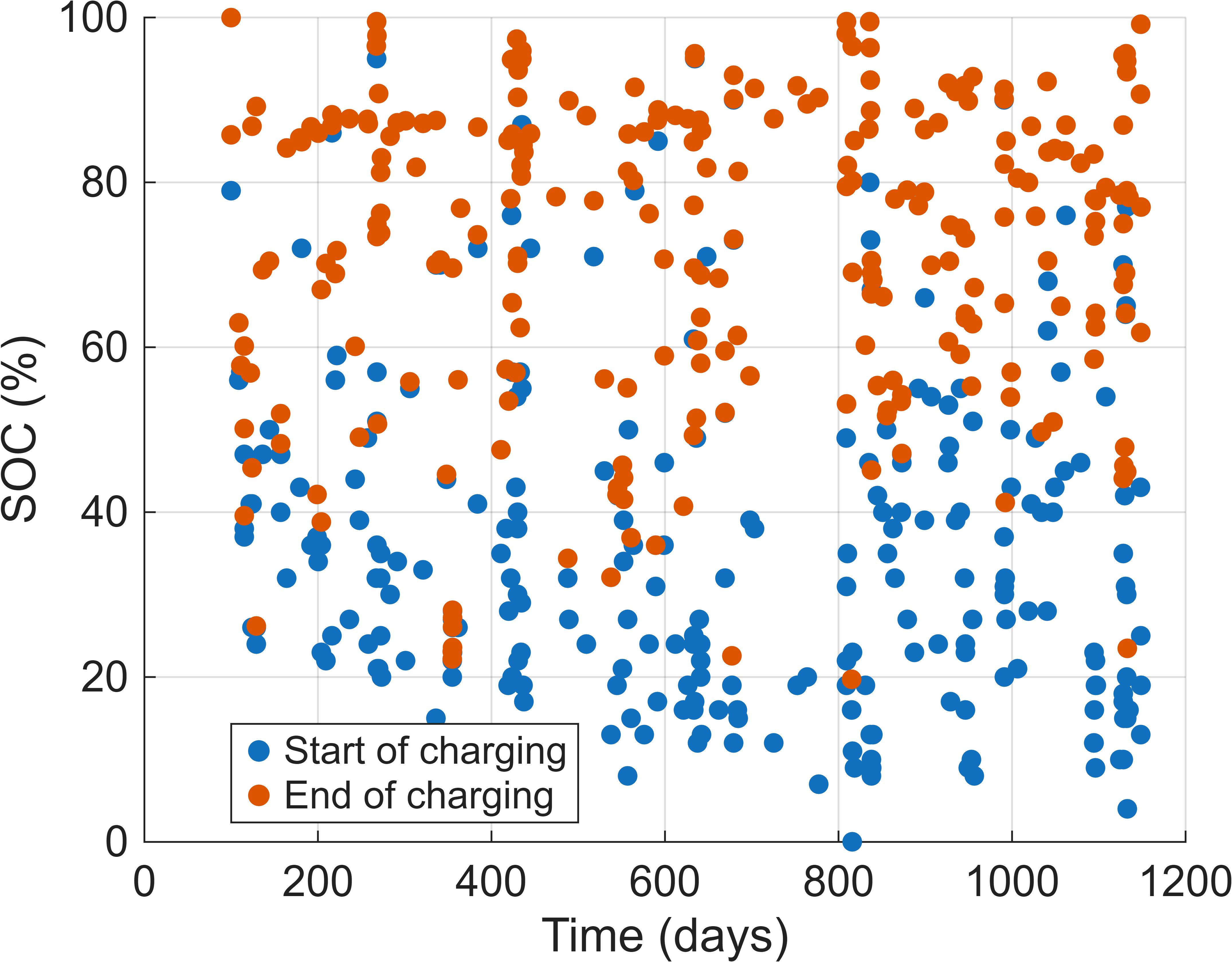}
\subcaption{The charging records used for the long-term aging simulation.}
\end{subfigure}
\begin{subfigure}{0.56\textwidth}
\includegraphics[width=\textwidth]{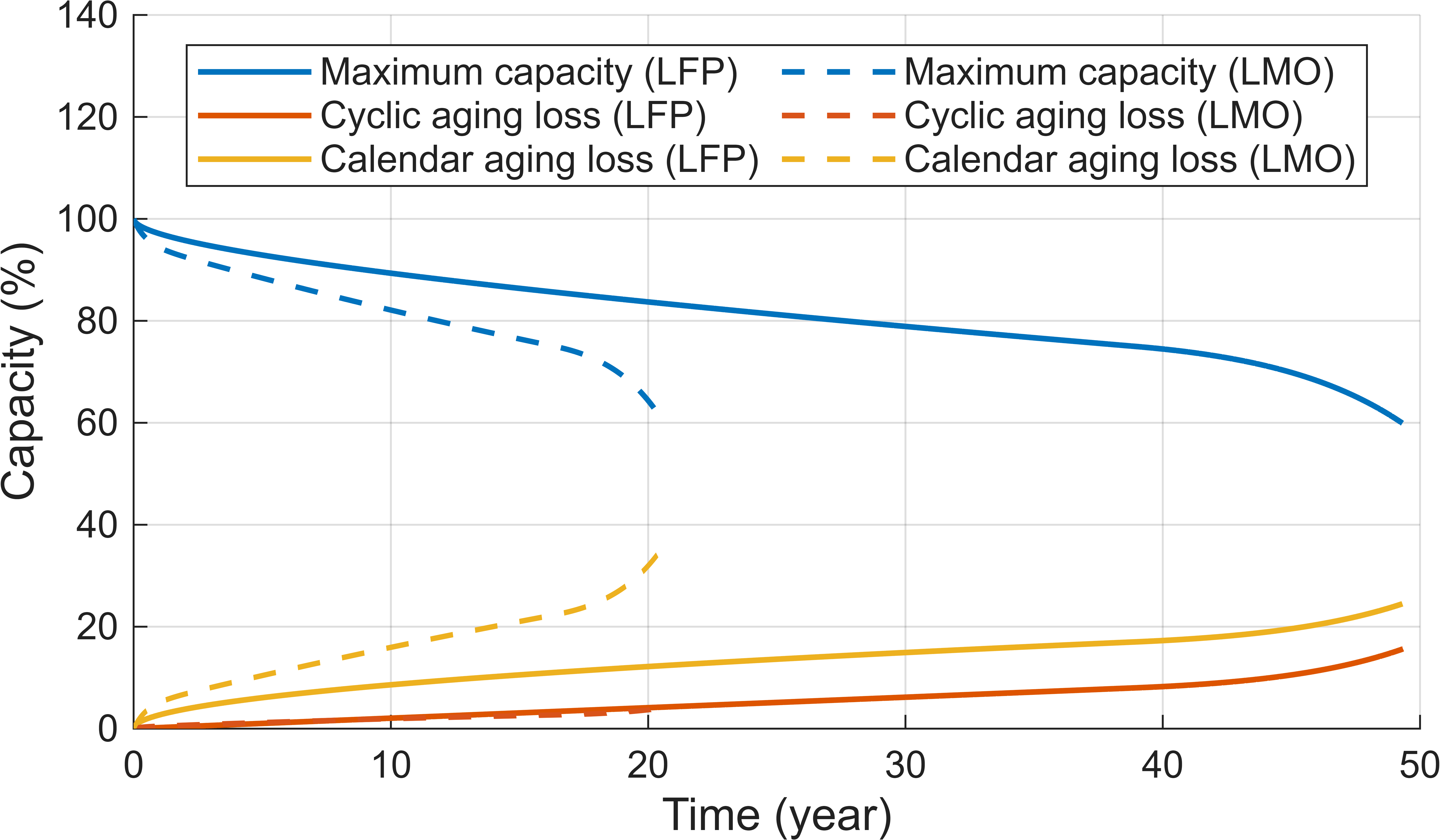}
\subcaption{Typical aging curves of two types of lithium-ion batteries.}
\end{subfigure}
\caption{Long-term aging simulation setups.}
\label{aging_exp}
\end{figure}

Note that, in both simulations, the proposed strategy balances SOH only within each phase. Therefore, for brevity, we report results only for the first phase of the three. 

The default specifications of the battery pack and chargers used in the simulations are as follows. Note that these specifications are not fixed across all cases; the impact of varying them is discussed in Section~\ref{sec_sensitive}. By default, each cell has a nominal capacity of 2.3~Ah, and each phase comprises 20 cells. The internal resistance of each cell is $R_0=0.01\,\Omega$. During relaxation, the surface temperature of each cell is assumed to be \(25^{\circ}\mathrm{C}\). During driving and charging, the surface temperature is modeled as a cell-to-cell random variation: for each cell, \(T\) is sampled once from a Gaussian distribution with mean \(35^{\circ}\mathrm{C}\) and standard deviation \(2^{\circ}\mathrm{C}\), and the sampled value is held constant over time. During charging, the maximum allowable C-rate \(I^{no}_{\max}\) is modeled as a function of SOC. By fitting the charging data of a Tesla Model 3 available on EVKX.net, we obtain the following empirical function:
\begin{equation}
    I^{no}_{\max}=2.6963-2.5795\cdot SOC
\end{equation}

 Two lithium-ion chemistries are considered: lithium iron phosphate (LFP) and lithium manganese oxide (LMO). The OCV--SOC relationships for LFP~\cite{lfpocv} and LMO~\cite{lmoocv} are given in \eqref{ocv_lfp} and \eqref{ocv_lmo}, respectively, where \(SOC\in[0,1]\).
\small{\begin{equation}\label{ocv_lfp}
    OCV_{\mathrm{LFP}} = -0.5863 \exp(-21.9\,SOC) + 3.414 + 0.1102\,SOC - 0.1718\,\exp\!\Big(\frac{-0.008}{1 - SOC}\Big)
\end{equation}}
\small{\begin{equation}\label{ocv_lmo}
    OCV_{\mathrm{LMO}} = 3.875 - 0.335 \big(-\ln(SOC)\big)^{0.653} - 0.5332\,SOC + 0.8315 \exp\!\big(0.6(SOC-1)\big)
\end{equation}}
\normalsize 
The aging models for the two chemistries are adopted from \cite{agingmodel,agingmodel2}. In these models, calendar aging depends on \(SOC\), temperature, and storage time, whereas cyclic aging depends on temperature, C-rate, Ah throughput, depth of discharge, and average \(SOC\). For LFP cells, \cite{agingmodel} models the short-term degradation as
\small{\begin{equation}\label{equ_LFP}
    \Delta SOH_{\mathrm{LFP}} =
    (aT^2+bT+c)\exp\!\big((dT+e) I^{no}\big)\cdot \frac{|\Delta SOC|}{4.6}
    + \frac{f}{2.3}\exp\!\big(g \overline{SOC} + h/T\big)\cdot \frac{\Delta t}{2\sqrt{t}},
\end{equation}}
where \(T\) is the temperature in Kelvin, \(I^{no}\) is the average C-rate, \(\overline{SOC}\) is the average \(SOC\), \(\Delta SOC\) is the \(SOC\) change over the (dis)charge interval, \(t\) is the accumulated calendar-aging time, and \(\Delta t\) is the duration of the current interval. The parameters \(a\)–\(h\) are given in Table~\ref{aging_coe}. During relaxation, \(\Delta SOC=0\) (and \(I^{no}=0\)), so \eqref{equ_LFP} reduces to the calendar-aging term (the second term).

For LMO cells, \cite{agingmodel2} models the short-term degradation as
\scriptsize{\begin{equation}\label{equ_LMO}
    \Delta SOH_{\mathrm{LMO}} =
    \Big[\alpha_{\mathrm{sei}}\beta_{\mathrm{sei}} \exp\!\big(-\beta_{\mathrm{sei}} (f_t+\sum_{i=1}^N f_{c,i})\big)
    +(1-\alpha_{\mathrm{sei}})\exp\!\big(-(f_t+\sum_{i=1}^N f_{c,i})\big)\Big]
    \Big(f_{c,N+1}+\frac{f_t\Delta t}{t}\Big),
\end{equation}}
\normalsize
where \(N\) is the total number of half cycles completed, and \(\alpha_{\mathrm{sei}}\) and \(\beta_{\mathrm{sei}}\) are model parameters. The functions \(f_{c,i}\) and \(f_t\) are
\begin{equation}
    f_{c,i}=\frac{0.5}{k_{\delta 1}|\Delta SOC_i|^{k_{\delta 2}}+k_{\delta 3}}
    \exp\!\Big(k_\sigma(\overline{SOC}_i-\sigma_{ref})+\frac{k_TT_{ref}(T-T_{ref})}{T}\Big),
\end{equation}
\begin{equation}
    f_t=k_t\,t\,
    \exp\!\Big(k_\sigma(\overline{SOC}-\sigma_{ref})+\frac{k_TT_{ref}(T-T_{ref})}{T}\Big),
\end{equation}
where \(\Delta SOC_i\) and \(\overline{SOC}_i\) denote the \(SOC\) swing and average \(SOC\) of the \(i\)th half cycle. In addition, \(\Delta SOC_{N+1}=\Delta SOC\) and \(\overline{SOC}_{N+1}=\overline{SOC}\). The constants \(k_{\delta 1},k_{\delta 2},k_{\delta 3},k_{\sigma}, \sigma_{ref}, k_T, T_{ref}\), and \(k_t\) are listed in Table~\ref{aging_coe}.

\begin{table}[htbp]
\caption{Parameters of the battery aging models}\footnotesize
\centering
\begin{tabular}{|c|c|c|c|c|c|}
\hline
\textbf{Parameter} & \boldmath$a$ & \boldmath{$b$} & \boldmath{$c$} & \boldmath{$d$} & \boldmath{$e$} \\ \hline
\textbf{Values} & $2.0916\text{E}{-8}$ & $-1.2179\text{E}{-5}$ & 0.0018 & $-1.7082\text{E}{-6}$ & $0.0556$ \\ \hline
\textbf{Parameter} & \boldmath{$f$} & \boldmath{$g$} & \boldmath{$h$} & \boldmath{} & \multicolumn{1}{l|}{} \\ \hline
\textbf{Values} & $5.9808\text{E}{6}$ & $0.6898$ & $-6.4647\text{E}{3}$ &  & \multicolumn{1}{l|}{} \\ \hline
\textbf{Parameter} & \boldmath$\alpha_{\text{sei}}$ & \boldmath$\beta_{\text{sei}}$ & \textbf{\boldmath$k_{\delta_1}$} & \textbf{\boldmath$k_{\delta_2}$} & \boldmath$k_{\delta_3}$ \\ \hline
\textbf{Values} & $5.75\text{E}{-2}$ & 121 & $1.40\text{E}{5}$ & $-5.01\text{E}{-1}$ & $-1.23\text{E}{5}$ \\ \hline
\textbf{Parameter} & \boldmath$k_{\sigma}$ & \boldmath$\sigma_{\text{ref}}$ & \boldmath$k_{T}$ & \boldmath$T_{\text{ref}}$ & \boldmath$k_{t}$ \\ \hline
\textbf{Values} & $1.04$ & $0.50$ & $6.93\text{E}{-2}$ & $298.15\text{ K}$ & $4.14\text{E}{-10}$ s$^{-1}$ \\ \hline
\end{tabular}
\label{aging_coe}
\end{table}

However, the aging models in \eqref{equ_LFP} and \eqref{equ_LMO} have two limitations in the present context. First, they do not account for cell-to-cell heterogeneity within a pack. Second, they are primarily calibrated to the early-life regime (approximately 80--100\,\% SOH), where the degradation rate initially decreases and then becomes nearly constant (i.e., the SOH trajectory is typically convex in cycle/time). In contrast, many experimental studies report a knee point in lithium-ion aging curves \cite{knee1}, after which the degradation rate accelerates markedly (i.e., the trajectory becomes concave). Therefore, to simulate the degradation of individual cells, we augment the baseline model by introducing (i) a cell-specific multiplicative factor to represent heterogeneity and (ii) an additional acceleration term that activates once SOH falls below a prescribed knee point. Specifically, the short-term degradation of cell \(i\) is computed by \eqref{lfp_equ2} or \eqref{lmo_equ2}, depending on the chemistry. The aging model is used only to evaluate battery degradation \emph{after} optimization; it is not embedded in the optimization to preserve the general applicability of the proposed method.

\begin{equation} \label{lfp_equ2}
    \Delta \mathrm{SOH}_{\mathrm{LFP},i}
    = \gamma_i \Big[1+50\max\{\mathrm{SOH}_{\mathrm{knee}}-\mathrm{SOH}_{\mathrm{LFP},i},\,0\}\Big]\cdot \Delta \mathrm{SOH}_{\mathrm{LFP}}
\end{equation}
\begin{equation} \label{lmo_equ2}
    \Delta \mathrm{SOH}_{\mathrm{LMO},i}
    = \gamma_i \Big[1+50\max\{\mathrm{SOH}_{\mathrm{knee}}-\mathrm{SOH}_{\mathrm{LMO},i},\,0\}\Big]\cdot \Delta \mathrm{SOH}_{\mathrm{LMO}}
\end{equation}
where \(\gamma_i \sim \mathcal N(1,\sigma_\gamma^2)\) is a cell-specific random factor (sampled once per cell and held constant over time) that captures cell-to-cell variation. By default, \(\sigma_\gamma=0.1\), corresponding to a 10\,\% RMS variation in aging rate across cells under identical operating conditions. Variables \(\mathrm{SOH}_{\mathrm{LFP},i}\) and \(\mathrm{SOH}_{\mathrm{LMO},i}\) denote the current SOH of cell \(i\), and \(\mathrm{SOH}_{\mathrm{knee}}\) denotes the knee-point SOH, set to 75\,\% by default. Figure~\ref{aging_exp}(b) illustrates typical simulated aging trajectories for LFP and LMO cells under representative EV usage. Under the assumed usage pattern with long dwell times (i.e., low cycling frequency), calendar aging contributes a substantial fraction of the total degradation.

\subsection{Simulation results}

The drive-cycle results are shown in Fig.~\ref{onecycle}(a) and (b). For visualization, each phase contains 10 LFP cells, ranked by SOH (cell \(i\) has the \(i\)th highest SOH). In both cases, the pack starts at \(0\,\%\) SOC, is charged to \(70\,\%\), and is then fully discharged. Under the charging-time and pack-SOC constraints, the proposed strategy allocates different charging amounts to different cells. The key difference between the two cases is the available charging time: Fig.~\ref{onecycle}(a) uses a shorter charging window than Fig.~\ref{onecycle}(b). As shown, cells with higher SOH are charged more than those with lower SOH, which promotes long-term SOH balancing. Moreover, the remaining capacity of all cells is balanced within the first 25\% of the discharge, indicating that the strategy does not reduce the usable energy (and thus range) in this example.

\begin{figure}[htbp]
\centering
\begin{subfigure}{0.48\textwidth}
\includegraphics[width=\textwidth]{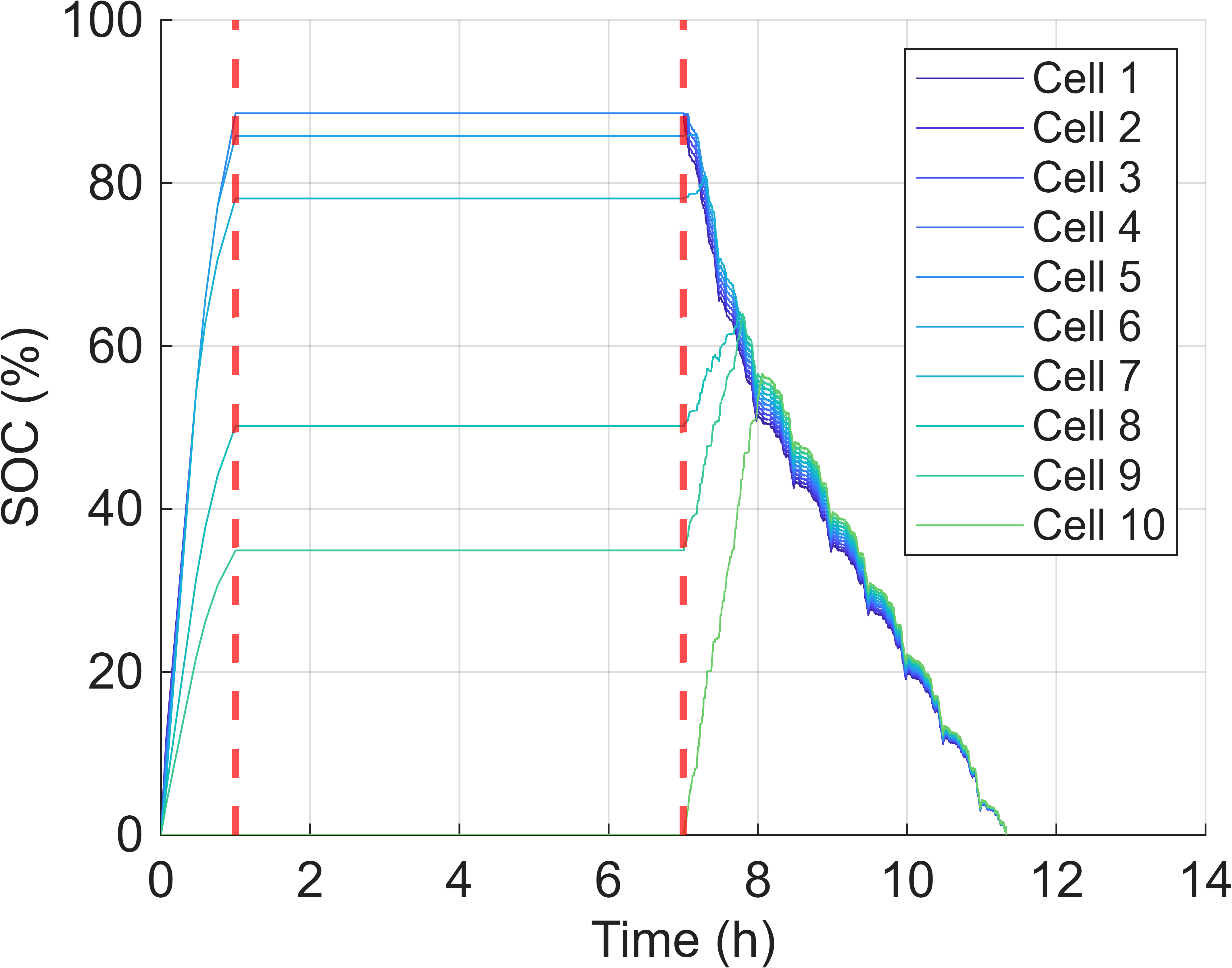}
\subcaption{1 h charging + 6 h rest + discharging.}
\end{subfigure}
\begin{subfigure}{0.48\textwidth}
\includegraphics[width=\textwidth]{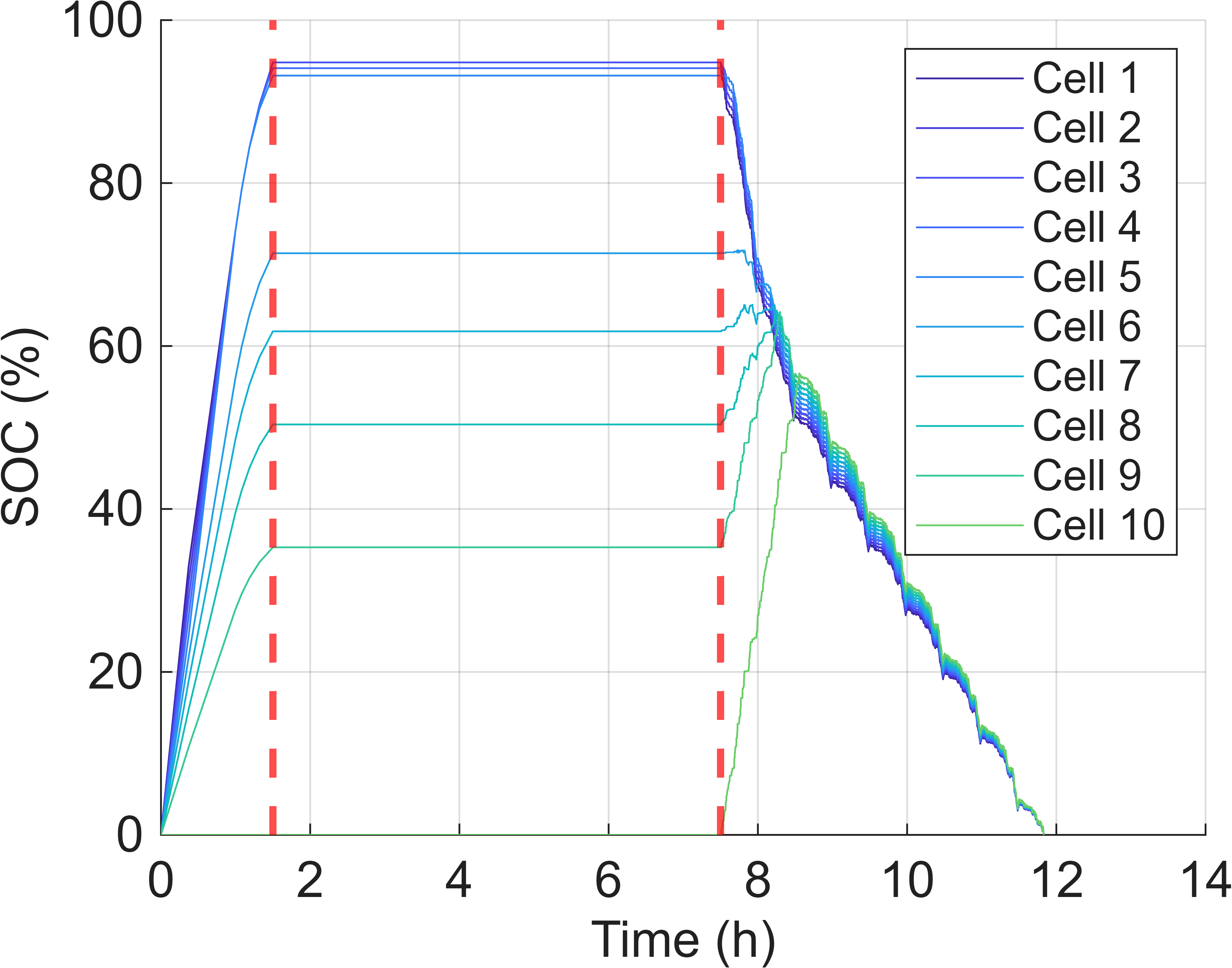}
\subcaption{1.5 h charging + 6 h rest + discharging.}
\end{subfigure}
\caption{Examples of the proposed charging \& discharging strategy. Cells with higher SOH are used more frequently than those with lower SOH, thereby balancing SOH and prolonging the battery pack's lifetime.}
\label{onecycle}
\end{figure}

We also observe that the SOC variance across cells at the end of charging is larger when the charging time is longer. This occurs because the CV stage is relatively slow: with a longer charging window, the controller can spend more time in the CV regime and push healthier cells to higher SOC while bypassing less healthy cells. Since the target pack SOC is fixed, this increases the cell-to-cell SOC dispersion and, in turn, enables more aggressive SOH balancing when more charging time is available.

Another observation from Fig.~\ref{onecycle} is that several cells share identical SOC trajectories during charging. This is a consequence of the CV voltage limit: the pack charging current is constrained by the cell that reaches the voltage limit first (typically the cell with the highest SOC), which encourages the controller to equalize SOC among a subset of cells to increase the allowable current and accelerate charging. In contrast, during discharge, the strategy balances remaining capacity, so cells tend to maintain similar remaining capacities rather than identical SOC values.

In the long-term aging simulations, the baseline controller is the SOC-balancing strategy (Strategy \ref{greedy3} in Section \ref{section_greedy}), which equalizes cell SOC by assigning cells with higher SOC to lower voltage levels during charging and higher voltage levels during discharging. We do not use SOH balancing as a baseline because it cannot fully utilize the battery pack’s energy, thereby rendering some discharge profiles infeasible. For a fair comparison, both the baseline and the proposed strategy are evaluated on the same cell-level inverter under identical specifications. Since the proposed strategy relies on per-cell SOH information and practical SOH estimators are imperfect, we also simulate a noisy-SOH case in which each cell's SOH estimate is corrupted by additive zero-mean Gaussian noise with a standard deviation of 2\,\%. The results for LFP and LMO chemistries are shown in Fig.~\ref{result_all}(a) and (b), respectively. Overall, the proposed strategy improves SOH balance across cells and extends pack lifetime by approximately 11--18\,\% relative to SOC balancing, even under noisy SOH estimates.

\begin{figure}[htbp]
    \centering
    %\begin{minipage}[htbp]{\textwidth}
        %\centering
        \begin{subfigure}[htbp]{\textwidth}
            \centering
            \includegraphics[width=0.95\textwidth]{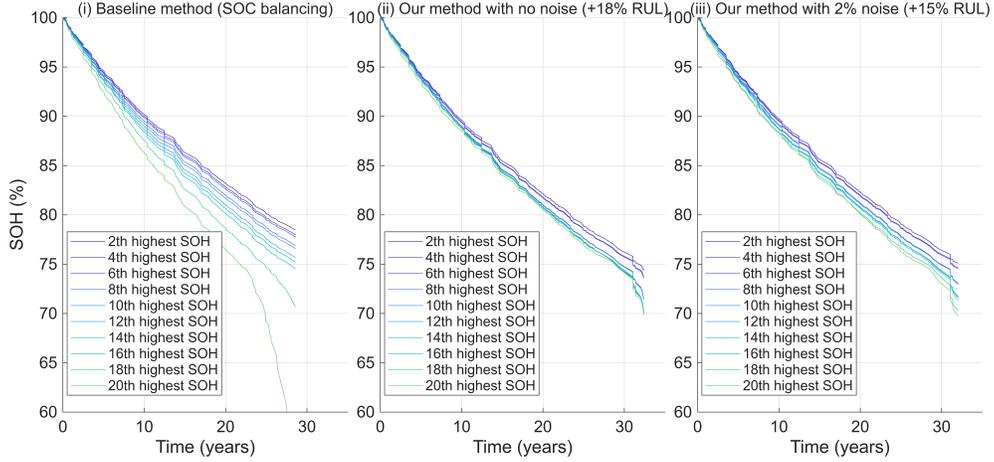} % Replace with your image file
            \caption{LFP battery}
        \end{subfigure}
        
        \vspace{0.4cm} % Adjust space between subfigures if needed
        
        \begin{subfigure}[htbp]{\textwidth}
            \centering
            \includegraphics[width=0.95\textwidth]{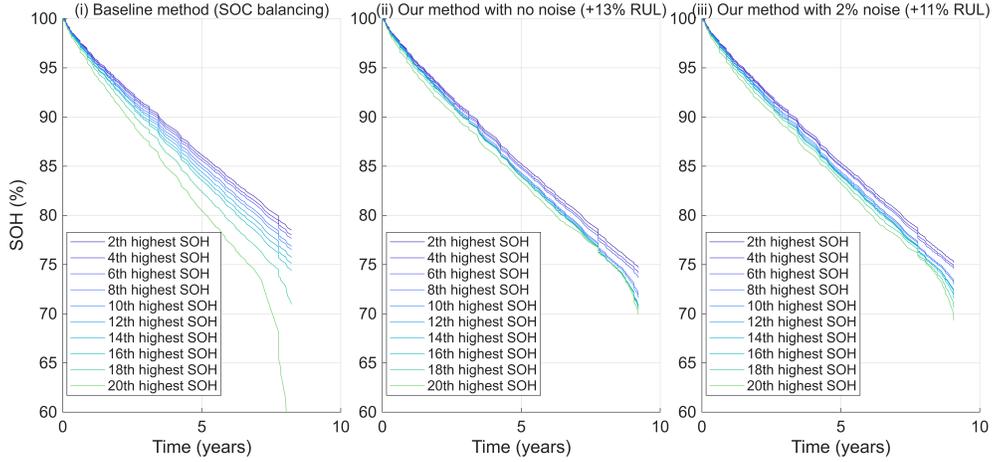} % Replace with your image file
            \caption{LMO battery}
        \end{subfigure}
    %\end{minipage}
    \caption{Comparison of different strategies in the long-term aging simulation for (a) LFP and (b) LMO battery chemistries. The proposed SOC–SOH-aware control strategy increases pack lifetime by up to 18\% and 13\% for LFP and LMO, respectively, relative to SOC balancing.}
    \label{result_all}
\end{figure}
\subsection{Sensitivity analysis}\label{sec_sensitive}
To further assess robustness, we repeat the long-term aging simulations under thirteen additional scenarios. The results are summarized in Table~\ref{analysis}, showing lifetime improvements of 7--38\,\% across a wide range of use cases and model variations.

\begin{table}[htbp]\caption{Battery lifetime improvement of the proposed SOC–SOH-aware control strategy relative to SOC-only balancing (Strategy \ref{greedy3} in Section \ref{section_greedy}) across operating conditions. The ``w/ noise'' column represents the case in which each cell's SOH estimate is corrupted by additive zero-mean Gaussian noise with a standard deviation of 2\,\%. Average changes with respect to the default scenario are presented in the final column.}\label{analysis} \tiny
\centering
\begin{tabular}{|c|c|cc|cc|c|}
\hline
 &  & \multicolumn{2}{c|}{\textbf{LFP}} & \multicolumn{2}{c|}{\textbf{LMO}} &  \\ \cline{3-6}
 & \multirow{-2}{*}{\textbf{Scenarios}} & \multicolumn{1}{c|}{w/o noise} & w/ noise & \multicolumn{1}{c|}{w/o error} & w/ error & \multirow{-2}{*}{\begin{tabular}[c]{@{}c@{}}\textbf{Average}\\ \textbf{effect}\end{tabular}} \\ \cline{2-7}
\multirow{-3}{*}{} & Default & \multicolumn{1}{c|}{18.1\,\%} & 14.9\,\% & \multicolumn{1}{c|}{12.8\,\%} & 10.9\,\% & \boldmath{$0.0\,\%$} \\ \hline

 & Doubled calendar aging & \multicolumn{1}{c|}{7.5\,\%} & 7.5\,\% & \multicolumn{1}{c|}{8.5\,\%} & 8.5\,\% & \boldmath{$-6.2\,\%$} \\ \cline{2-7}

 & Doubled cyclic aging & \multicolumn{1}{c|}{18.2\,\%} & 16.1\,\% & \multicolumn{1}{c|}{11.5\,\%} & 10.4\,\% & \boldmath{$-0.1\,\%$} \\ \cline{2-7}

 & \begin{tabular}[c]{@{}c@{}}Earlier knee point\\ ($SOH_{knee}=85\,\%$)\end{tabular} & \multicolumn{1}{c|}{18.5\,\%} & 14.5\,\% & \multicolumn{1}{c|}{11.1\,\%} & 7.0\,\% & \boldmath{$-1.4\,\%$} \\ \cline{2-7}

\multirow{-5}{*}{\begin{tabular}[c]{@{}c@{}}Different\\ aging\\ models\end{tabular}} & \begin{tabular}[c]{@{}c@{}}No knee point\\ ($SOH_{knee}=0$)\end{tabular} & \multicolumn{1}{c|}{16.6\,\%} & 12.7\,\% & \multicolumn{1}{c|}{13.1\,\%} & 11.3\,\% & \boldmath{$-0.8\,\%$} \\ \hline

 & \begin{tabular}[c]{@{}c@{}}More cells (40 cells)\end{tabular} & \multicolumn{1}{c|}{18.3\,\%} & 17.8\,\% & \multicolumn{1}{c|}{12.8\,\%} & 12.1\,\% & \boldmath{$+1.1\,\%$} \\ \cline{2-7}

 & \begin{tabular}[c]{@{}c@{}}More diverse cells\\ (standard deviation of \\ $\gamma_i$ increased to 0.15)\end{tabular} & \multicolumn{1}{c|}{21.3\,\%} & 15.9\,\% & \multicolumn{1}{c|}{13.8\,\%} & 11.1\,\% & \boldmath{$+1.4\,\%$} \\ \cline{2-7}

 & \begin{tabular}[c]{@{}c@{}}Larger internal resistance\\ ($R_0$ increased to $0.1\,\Omega$)\end{tabular} & \multicolumn{1}{c|}{18.1\,\%} & 14.9\,\% & \multicolumn{1}{c|}{12.8\,\%} & 10.9\,\% & \boldmath{$0.0\,\%$} \\ \cline{2-7}

 & \begin{tabular}[c]{@{}c@{}}More fast charging sessions\\ (proportion of fast charging\\ increased to 50\,\%)\end{tabular} & \multicolumn{1}{c|}{12.1\,\%} & 10.1\,\% & \multicolumn{1}{c|}{8.7\,\%} & 8.0\,\% & \boldmath{$-4.5\,\%$} \\ \cline{2-7}

 & \begin{tabular}[c]{@{}c@{}}Higher temperature\\ (average charge/discharge\\ temperature increased \\ to $45\,^{\circ}\text{C}$)\end{tabular} & \multicolumn{1}{c|}{18.6\,\%} & 13.2\,\% & \multicolumn{1}{c|}{12.4\,\%} & 10.9\,\% & \boldmath{$-0.4\,\%$} \\ \cline{2-7}

\multirow{-19}{*}{\begin{tabular}[c]{@{}c@{}}Different \\ use \\ cases\end{tabular}} & \begin{tabular}[c]{@{}c@{}}More diverse temperature\\ (standard deviation \\ increased to $4\,^{\circ}\text{C}$)\end{tabular} & \multicolumn{1}{c|}{18.1\,\%} & 14.9\,\% & \multicolumn{1}{c|}{12.8\,\%} & 10.9\,\% & \boldmath{$0.0\,\%$} \\ \hline

 & $SOH_\mathrm{pack}=80\,\%$ & \multicolumn{1}{c|}{38.0\,\%} & 32.6\,\% & \multicolumn{1}{c|}{22.4\,\%} & 17.9\,\% & \boldmath{$+13.6\,\%$} \\ \cline{2-7}

 & $SOH_\mathrm{pack}=75\,\%$ & \multicolumn{1}{c|}{21.4\,\%} & 16.3\,\% & \multicolumn{1}{c|}{15.1\,\%} & 13.1\,\% & \boldmath{$+2.3\,\%$} \\ \cline{2-7}

\multirow{-3}{*}{\begin{tabular}[c]{@{}c@{}}Different\\ definitions\\ of EOL\end{tabular}} & $SOH_\mathrm{pack}=65\,\%$ & \multicolumn{1}{c|}{15.1\,\%} & 14.1\,\% & \multicolumn{1}{c|}{12.5\,\%} & 8.1\,\% & \boldmath{$-1.7\,\%$} \\ \hline
\end{tabular}
\end{table}

According to Table~\ref{analysis}, the scenarios that most noticeably reduce the benefit of the proposed strategy are increasing the proportion of DC fast charging to 50\,\% and doubling the calendar-aging capacity loss. The former is expected because the simulation assumes little or no SOH balancing during fast charging in order to prioritize charging speed. The latter is also expected because the proposed strategy is optimization-based, and its objective in \eqref{cf} is influenced more directly by cycling-induced degradation than by calendar aging.

In contrast, the scenario that most notably increases the benefit of the proposed strategy is setting the battery-pack EOL threshold to \(SOH_\mathrm{pack}=80\,\%\). The reason is as follows. As explained in Section~\ref{end_pack}, cells with \(SOH<SOH_\mathrm{pack}\) can no longer be used in the pack; therefore, the pack SOH drops abruptly whenever a cell reaches \(SOH_\mathrm{pack}\). When \(SOH_\mathrm{pack}=70\,\%\), the pack typically reaches EOL when approximately 10\,\% of the cells have reached \(SOH_\mathrm{pack}\). By contrast, when \(SOH_\mathrm{pack}=80\,\%\), the pack typically reaches EOL when only about 5\,\% of the cells have reached \(SOH_\mathrm{pack}\). Under the proposed strategy, the cell SOHs remain relatively uniform, so the time at which 5\,\% of cells reach \(SOH_\mathrm{pack}\) is close to the time at which 10\,\% reach \(SOH_\mathrm{pack}\). Under the conventional strategy, however, this time gap can be much larger. Therefore, increasing \(SOH_\mathrm{pack}\) increases the benefit of the proposed strategy.

\section{Conclusions}
In this paper, we propose a SOC–SOH-aware control strategy to extend the range and service life of EV battery packs equipped with a cell-level inverter powertrain. We focus on applying the strategy during charging, particularly during AC charging events. The strategy does not require an explicit aging model and is robust to measurement noise and SOH estimation errors. We evaluate the strategy on two lithium-ion chemistries under a range of operating conditions. The results show that the proposed strategy can extend pack lifetime by 7--38\,\% relative to a commonly adopted SOC-balancing strategy, while maintaining comparable usable energy for subsequent driving.

One limitation of the proposed strategy is that its quantitative benefit depends on the underlying aging characteristics. As illustrated in Fig.~\ref{onecycle}, SOH balancing is achieved by charging different cells to different SOC levels. If the degradation rate is strongly convex in SOC over the relevant operating range, increasing SOC dispersion can raise the average degradation rate (all else equal), thereby reducing the net benefit of SOH balancing. This limitation could be mitigated by incorporating available aging-model knowledge (or learned degradation surrogates) into the charging optimization, which is a promising direction for future work. 
%Future studies may also extend the strategy to other powertrain topologies, such as \cite{topofuture}, which supports both series and parallel connections.
\ifblindreview
\section*{CRediT authorship contribution statement}
CRediT authorship contribution statement is withheld for double-blind review.
\else
  \section*{CRediT authorship contribution statement}
\textbf{Shida Jiang:} Conceptualization, Investigation, Formal analysis, Visualization, Methodology, Validation, Writing – original draft, Writing – review \& editing. \textbf{Shengyu Tao:} Formal analysis, Visualization, Writing – review \& editing.\textbf{Vincent Molina:} Data curation, Supervision. \textbf{Junzhe Shi:} Methodology, Writing – review \& editing. \textbf{Scott Moura:} Conceptualization, Data curation, Funding acquisition, Project administration, Writing – review \& editing.
\fi

\section*{Declaration of competing interest}
The authors declare that they have no known competing financial interests or personal relationships that could have appeared to influence the work reported in this paper.

\ifblindreview
  \section*{Acknowledgement}
  Acknowledgements are withheld for double-blind review.
\else
  \section*{Acknowledgement}
  This work was funded by BMW USA. The authors also thank Markus Anton Preisinger and Hadzalic Nejira, who were both interns at BMW, for their support and valuable feedback during the course of this project.
\fi

\section*{Data availability}
Data will be made available on request.

%% The Appendices part is started with the command \appendix;
%% appendix sections are then done as normal sections
%% \appendix

%% \section{}
%% \label{}

%% If you have bibdatabase file and want bibtex to generate the
%% bibitems, please use
%%
\bibliographystyle{elsarticle-num} 
\bibliography{elsarticle-template-num}

%% else use the following coding to input the bibitems directly in the
%% TeX file.

%\begin{thebibliography}{00}

%% \bibitem{label}
%% Text of bibliographic item

%\bibitem{}

%\end{thebibliography}

%\section{Confidence interval of the SOC and SOH estimation}

\end{document}